\documentclass[11pt]{article}
\usepackage{jheppub}
\usepackage{amsmath,amssymb,amsfonts,graphicx}
\usepackage{braket}
\usepackage{bm}
\usepackage{amsthm}
\usepackage[caption=false]{subfig}

\makeatletter
\def\@fpheader{\relax}
\makeatother

\newcommand{\be}{\begin{equation}}
\newcommand{\ee}{\end{equation}}
\newcommand{\bea}{\begin{eqnarray}}
\newcommand{\eea}{\end{eqnarray}}
\newcommand{\beas}{\begin{eqnarray*}}
\newcommand{\eeas}{\end{eqnarray*}}
\newcommand{\ba}{\begin{array}}
\newcommand{\ea}{\end{array}}

\renewcommand*\d[2][]{%
	\mathrm{d}%
	\ifx\relax#1\relax\else
	\rule{-0.02em}{1.5ex}^{#1}\rule{0.08em}{0ex}\!
	\fi
	#2\,
}

\newtheorem*{lemma}{Lemma}

\newcommand{\Dbar}{\bar \Delta}
\newcommand{\lam}{\lambda}
\newcommand{\lamt}{\Tilde{\lambda}}
\newcommand{\vt}{\Tilde{v}}
\newcommand{\ut}{\Tilde{u}}
\newcommand{\Vt}{\Tilde{V}}
\newcommand{\Ut}{\Tilde{U}}
\newcommand{\yt}{\Tilde{y}}

\newcommand{\um}{u^{(-)}}
\newcommand{\up}{u^{(+)}}
\newcommand{\utm}{\Tilde u^{-}}
\newcommand{\utp}{\Tilde u^{+}}
\newcommand{\jmax}{j_\text{max}}
\newcommand{\afrw}{\tilde{a}_\textsc{FRW}}

\title{Can one hear the shape of a wormhole?}

\author[1]{Stefano Antonini,}
\author[2]{Petar Simidzija,}
\author[1,3]{Brian Swingle,}
\author[2]{Mark Van Raamsdonk}

\affiliation[1]{Maryland Center for Fundamental Physics, University of Maryland, College Park, MD 20742, USA}
\affiliation[2]{Department of Physics and Astronomy, University of British Columbia, Vancouver, B.C.\ V6T 1Z1, Canada.}
 \affiliation[3]{Brandeis University, Waltham, MA 02453, USA}

\emailAdd{santonin@umd.edu}
\emailAdd{psimidzija@phas.ubc.ca}
\emailAdd{bswingle@brandeis.edu}
\emailAdd{mav@phas.ubc.ca}

\abstract{A large class of flat big bang - big crunch cosmologies with negative cosmological constant are related by analytic continuation to asymptotically AdS traversable wormholes with planar cross section. In recent works (arXiv: 2102.05057, 2203.11220) it was suggested that such wormhole geometries may be dual to a pair of 3D holographic CFTs coupled via auxiliary degrees of freedom to give a theory that confines in the infrared. In this paper, we explore signatures of the presence of such a wormhole in the state of the coupled pair of 3D theories. We explain how the wormhole geometry is reflected in the spectrum of the confining theory and the behavior of two-point functions and entanglement entropies. We provide explicit algorithms to reconstruct the wormhole scale factor (which uniquely determines its geometry) from entanglement entropies, heavy operator two-point functions, or light operator two-point functions (which contain the spectrum information). In the last case, the physics of the bulk scalar field dual to the light operator is closely related to the quantum mechanics of a one-dimensional particle in a potential derived from the scale factor, and the problem of reconstructing the scale factor from the two-point function is directly related to the problem of reconstructing this Schr\"odinger potential from its spectrum.}


\begin{document}

\maketitle
\newpage

\section{Introduction}

In this paper, we study the possible holographic description of Lorentzian geometries with two asymptotically AdS planar boundaries, with metric of the form
\begin{equation}
    ds^2=a^2(z)(dz^2-dt^2+dx^2+dy^2),
    \label{wormholemetric}
\end{equation}
where $z\in(-z_0,z_0)$ and $a(z)$ has simple poles at the locations $z=\pm z_0$ of the two asymptotic AdS boundaries. Such geometries make an intriguing appearance as the double analytic continuation of time-reversal symmetric $\Lambda < 0$ big-bang big crunch cosmologies. In \cite{VanRaamsdonk:2021qgv,Antonini2022,Antonini2022short}, following \cite{Maldacena:2004rf} we argued that understanding the holographic description of such wormholes may lead to microscopic models of cosmological physics, perhaps even relevant to our universe.

In this paper, we will assume that gravitational solutions of this type exist\footnote{Planar traversable wormholes require an anomalously large amount of negative energy to exist \cite{Freivogel:2019lej,VanRaamsdonk:2021qgv,Antonini2022}. An example mechanism leading to such enhanced negative energy was given in \cite{May:2021xhz}, and additional evidence for the existence of the solutions of our interest was presented in \cite{Antonini2022}.} and can be described holographically, and we will attempt to understand the required properties of this holographic description. The pair of asymptotically AdS boundaries in the wormhole suggests a holographic description involving a pair of CFTs. These must be highly entangled with each other in the state corresponding to the wormhole since the two asymptotic regions are connected in the interior. Since it is possible to travel from one boundary to the other causally through the spacetime, these CFTs must also be interacting. 

Various recent works \cite{Gao2016,Maldacena:2017axo, Maldacena:2018lmt, Maldacena:2018gjk, Freivogel:2019lej} have suggested specific ways to obtain traversable wormholes by coupling a pair of holographic CFTs. In \cite{VanRaamsdonk:2021qgv, Antonini2022} it was proposed that four-dimensional examples with geometry (\ref{wormholemetric}) might arise by coupling a pair of 3D CFTs via an auxiliary four-dimensional field theory. The 4D theory has many fewer local degrees of freedom than the 3D theories, but (via a renormalization group flow) strongly modifies the IR physics so that the IR theory is confining, with a ground state in which the two 3D CFTs are strongly entangled.

In this paper, we will be somewhat agnostic about how the CFTs associated with the two asymptotically AdS regions are coupled and ask instead how the wormhole physics is reflected in the state of the CFT degrees of freedom. We ask two general questions about the relationship between planar traversable wormhole geometries and observable properties in the dual field theory.

\paragraph{Question 1:} \textit{Given a microscopic setup, what features of the observables signal the presence of a dual eternal traversable wormhole?}

This first question is the subject of discussion in Section \ref{sec:CFT_from_wormhole}. We identify multiple signatures of the presence of a wormhole on the observables of the dual microscopic theory. 

First, the theory has a discrete mass spectrum of particles characteristic of a confining gauge theory. In the wormhole geometry, the two asymptotic boundaries correspond to the UV of the field theory, while the middle of the wormhole corresponds to an IR ``end'' of the geometry that is a finite distance from any interior point. Having an IR end characterizes a geometry whose dual field theory is confining.\footnote{See e.g. \cite{Witten1998a,Sakai:2004cn} for previous models of holographic confining gauge theories. In Witten's original model, we have an internal $S^1$ contracting smoothly to zero at the IR end; in the wormhole geometry, we have an $S^0$ that contracts smoothly to zero.}
With this feature, the various bulk fields exhibit a discrete set of modes that correspond to towers of particles (glueballs, etc...) with a discrete set of masses in the dual field theory. 

In sections \ref{sec:particle_spectrum} and \ref{sec:gauge_field}, we study in detail the spectra that arise from scalar fields and gauge fields in the wormhole and how these depend on the scale factor $a(z)$. Scalar fields in the wormhole geometry with normalizable boundary conditions can be decomposed into components associated with a discrete set of radial wavefunctions $u_i(z)$. These solve an auxiliary Schrodinger problem with potential
\be
\label{Schrodinger1}
V(z) = \frac{a''(z)}{a(z)} + m^2 a^2(z) \; .
\ee
The associated eigenvalues $\lambda_i$ give the values of the mass squared for the associated scalar particles in the 3D theory.  

When a $U(1)$ gauge field is present in the bulk theory, it gives rise to an evenly-spaced discrete mass spectrum of vector particles in the dual field theory as well as a massless sector. Depending on the boundary conditions for the gauge field, the massless sector is given either by a massless scalar field (which can be interpreted as the Goldstone boson for the spontaneous breaking of a global $U(1)\times U(1)$ symmetry down to a single $U(1)$ symmetry, induced by the coupling between the two 3D CFTs), or by a massless gauge field (associated to a gauged residual $U(1)$ symmetry in the 3D CFTs). The existence of the massless sector non-trivially implies the existence of long-range correlations in the vacuum state of our confining theory.

In Section \ref{sec:two_point_functions} we describe how the wormhole geometry implies a specific behavior for the two-point functions of scalar operators in the confining field theory, providing expressions for the one-sided and two-sided two-point functions in terms of the wormhole scale factor. 

The wormhole geometry also has implications for the behavior of entanglement entropies of subregions of the microscopic theory, as we describe in Section \ref{sec:ententropy}. For example, considering the entanglement entropy for a region that includes a ball of radius $R$ in each CFT, we expect a phase transition as the radius is increased past some critical radius where the RT surface changes from being disconnected to being connected.

The second question that we ask is the following:

\paragraph{Question 2:}\textit{How can the wormhole geometry be reconstructed from microscopic observables?}

We address this question in Section \ref{sec:wormhole_from_CFT}, where we identify three different observables from which to retrieve the behavior of the scale factor $a(z)$ appearing in equation (\ref{wormholemetric}). 

In Sections \ref{sec:spectrum_from_correlator} -\ref{sec:reconstruction_example} we show how to reconstruct the scale factor using the mass spectra of scalar particles in the confining gauge theory arising from scalar fields in the bulk geometry. The spectrum associated with a given scalar field can be extracted from the two-point function of the corresponding scalar operator, as explained in Section \ref{sec:spectrum_from_correlator}. Using the spectral information associated with a single massless scalar,\footnote{The spectrum from a single scalar with $m^2 < 0$ (and satisfying the Breitenlohner-Freedman bound) may also be sufficient, though we only have numerical evidence for this.} or any two scalar fields of arbitrary different masses\footnote{The spectrum of any scalar is enough to reconstruct the potential $V(z)$ in \ref{Schrodinger1}, but there can be multiple scale factors that result in this same potential. Having spectra from scalar fields with two different masses is always enough to fix $a(z)$ uniquely, though the ambiguity can likely be resolved with much less additional information, for example by looking at the short-distance behavior of a CFT two-point function or regularized entanglement entropy.}, we can reconstruct the scale factor.\footnote{The reconstruction is precise in the limit where the bulk theory is free.}
The problem of reconstructing the wormhole geometry from mass spectra reduces to the question of whether, given the discrete spectrum of a Schr\"odinger equation, we can reconstruct the potential $V(z)$ which generates it. This is a particular case of an inverse Sturm-Liouville problem \cite{Levitan-Gasymov,Hald}, a well-known mathematical question which, in its two-dimensional version, is frequently formulated as ``can one hear the shape of a drum?'' \cite{kac1966}. Although the answer to this question is negative for generic potentials, the symmetry of our configuration and the AdS asymptotics turn out to be enough to render the inverse problem solvable. We provide an explicit algorithm in Section \ref{sec:reconstruction_algorithm_summary}, first to reconstruct the Schr\"odinger potential appearing in (\ref{Schrodinger1}) and then using the potentials for a pair of different mass scalars (or a single potential for $m^2 \leq 0$) to reconstruct the scale factor. We give an explicit example in section \ref{sec:reconstruction_example}.

The second microscopic observable we can use to reconstruct the wormhole geometry is the two-point function of 3D CFT operators with large scaling dimension, corresponding to heavy bulk scalar fields (see Section \ref{sec:metric_from_heavy}). Such correlators can be evaluated in the geodesic approximation \cite{Faulkner:2018faa} leading to a functional dependence on the scale factor $a(z)$. Inverting this relationship by solving an integral equation allows us to reconstruct the wormhole metric in terms of microscopic correlators. 

The third observable we consider, in Section \ref{sec:metric_from_entanglement}, is the entanglement entropy of a strip-shaped subsystem in one of the two 3D CFTs. Similarly to the analysis for heavy correlators, the area of the Ryu-Takayangi (RT) surface \cite{Ryu2006} associated with such subregions has a functional dependence on the scale factor. Once inverted, this relationship yields the wormhole scale factor in terms of the entanglement entropy of the dual microscopic theory.

\subsubsection*{Applications to cosmology}

As we have discussed above, the holographic description of planar asymptotically AdS wormhole geometries studied in this paper is an essential part of the framework introduced in  \cite{VanRaamsdonk:2021qgv,Antonini2022,Antonini2022short} for describing certain time-symmetric $\Lambda < 0$ Big Bang-Big Crunch cosmologies holographically.
In that setup, the quantum state encoding the cosmological spacetime and the quantum state encoding the wormhole spacetime arise from two different slicings of the same Euclidean field theory path integral, and the observables in the two pictures are related by an analytic continuation of two spacetime coordinates. In particular, the cosmological scale factor in conformal time is the analytic continuation of the wormhole scale factor $a(z)$, and so a complete knowledge of the wormhole geometry implies a complete knowledge of the cosmological evolution. Therefore, the wormhole reconstruction procedures developed in Section \ref{sec:wormhole_from_CFT} imply that the cosmological scale factor can be reconstructed from simple confining field theory observables as well.\footnote{Probing the cosmological evolution without relying on the connection to the wormhole is a highly complex task \cite{VanRaamsdonk:2021qgv,Antonini2022,Cooper2018,Antonini2019} (see also \cite{Antonini:2021xar,Almheiri:2018ijj} for lower-dimensional examples, and \cite{Fallows:2022ioc,Waddell:2022fbn} for related discussions), so this approach gives an explicit example of the power of the ``slicing duality'' introduced in \cite{VanRaamsdonk:2021qgv,Antonini2022,Antonini2022short}} We will comment further on these aspects in Section \ref{sec:reconstruction_algorithm_summary} and in Section \ref{sec:discussion}, where we give our concluding remarks.

\section{From the wormhole to the microscopic theory}
\label{sec:CFT_from_wormhole}

Let us consider a planar traversable wormhole geometry 
\begin{align}\label{eq:wormhole}
    ds^2 = a^2(z)(-dt^2+dx^2+dy^2+dz^2),
\end{align}
with $t,x,y\in R$, and $z\in(-z_0,z_0)$ a coordinate with finite range. We assume that the wormhole connects two asymptotically AdS regions at $z=\pm z_0$, near which the scale factor $a(z)$ has the asymptotic form
\begin{align}\label{eq:a_asymp_ads}
    a(z) \sim \frac{L}{z_0\pm z} = \frac{\sqrt{-3/\Lambda}}{z_0\pm z}
\end{align}
where $\Lambda<0$ is the cosmological constant and $L$ is the AdS radius. We will also assume that the wormhole is symmetric, $a(z)=a(-z)$. As explained above and discussed in detail in \cite{Antonini2022,Antonini2022short}, such a geometry arises via a double analytic continuation from a large class of spatially flat, time-reversal symmetric Friedmann-Roberston-Walker (FRW) cosmologies with  $\Lambda < 0$.

Interestingly, while the state of the cosmological effective field theory in these models is highly excited and thus generally difficult to study, the corresponding state in the wormhole picture is simply the canonical vacuum state. This allows one to study a wide class of cosmological observables (e.g. density perturbations, the CMB spectrum, etc.) by studying the vacuum physics in the traversable wormhole background. Moreover, since the wormhole connects two asymptotically AdS regions, we expect the physics in the wormhole to be dual to a pair of non-gravitational, microscopic 3D CFTs living on the two AdS boundaries $\mathbb R^{1,2}$.\footnote{Since the two boundaries are causally connected in the bulk via the traversable wormhole, the microscopic 3D CFTs must be coupled in such a way that allows information to pass from one theory to the other. In \cite{VanRaamsdonk:2021qgv,Antonini2022} such a coupling is provided by a 4D theory with many fewer degrees of freedom, where the extra dimension is a compact interval. The resulting theory flows in the IR to a gapped 3D theory.} Our goal in this paper will be to better understand how the effective low energy physics in the wormhole is related to the underlying microscopic physics of the two coupled 3D CFTs. As such, this question does not refer to the cosmology picture, and thus could be of interest outside the cosmological context.

\subsection{Scalar particle spectrum from bulk scalar field}
\label{sec:particle_spectrum}

Let us begin by considering the simple situation in which the effective field theory in the wormhole consists of a single free scalar field $\phi$ of mass $m$.\footnote{This would be expected to hold precisely in a strict large $N$ limit of the dual field theory and approximately for large but finite $N$.} The vacuum of the scalar field corresponds to the vacuum state of the dual microscopic theory. We will compute the excitation spectrum of the scalar in the wormhole, and thus obtain the mass spectrum of particles in the microscopic theory. 

The excitation spectrum of the field $\phi$ on the wormhole background can be found by canonically quantizing $\phi$. We begin by finding a complete set of mode solutions to the wave equation $(\Box -m^2)\phi = 0$. Translation invariance in the $x, y$, and $t$ directions ensures a basis of solutions of the form
\begin{align}\label{eq:f_mode_functions}
    f_{k_x,k_y,i}(t,x,y,z) &= \frac{1}{2\pi\sqrt{2\omega}a(z)}e^{-i \omega t} e^{i k_x x} e^{i k_y y} u_i(z),
\end{align}
which, after substituting into the wave equation, gives a Schr\"odinger equation for the modes $u_i$, namely
\begin{align}\label{eq:SE}
    -u_i''(z)+ V(z)u_i(z)= \lambda_i u_i(z).
\end{align}
The frequencies $\omega$ are related to the transverse momenta $k_x,k_y$ and the eigenvalues $\lambda_i$ via
\begin{align}\label{eq:phi_energy_spectrum}
    \omega = \sqrt{k_x^2+k_y^2+\lambda_i},
\end{align}
and the Schr\"odinger potential $V(z)$ is defined in terms of the wormhole scale factor $a(z)$ by
\begin{align}\label{eq:V(z)}
    V(z) \equiv \frac{a''(z)}{a(z)}+m^2a^2(z).
\end{align}
Since the scale factor, by assumption, has the symmetry $a(z)=a(-z)$, the same is true for the potential, $V(z)=V(-z)$. Furthermore the scale factor has the asymptotic AdS form \eqref{eq:a_asymp_ads} near the boundaries, and so the potential is of the asymptotic form
\begin{align}\label{eq:V_asymptotics}
    V(z) &\sim \frac{\alpha}{(z_0\pm z)^2},\\
    \alpha &\equiv 2-3m^2/\Lambda.
\end{align}
From the Schr\"odinger equation we see that the two possibilities for the asymptotic behavior of $u(z)$ are either $u(z) \sim (z_0\mp z)^{\Delta_+}$ or $u(z) \sim (z_0\mp z)^{\Delta_-}$, where
\begin{align}
    \Delta_{\pm} &\equiv \frac{1\pm \sqrt{1+4\alpha}}{2}.
\end{align}
We call the former solutions \textit{normalizable}, and the latter \textit{non-normalizable}\footnote{Note that such denomination is meaningful only for $\alpha\geq 3/4$. For $-1/4<\alpha<3/4$, both solutions are normalizable in the $L^2$ norm, and there are two non-equivalent quantizations of the scalar field denominated ``standard'' and ``alternative'' quantizations, depending on whether we quantize the ``normalizable'' or ``non-normalizable'' part of the scalar field. When this ambiguity is present, we will focus here on the standard quantization scheme.}.

The canonical quantization of the scalar field is obtained by writing it as a mode sum
\begin{align}\label{eq:field_expansion}
    \hat{\phi}(t,x,y,z) = \sum_{i}\int dk_x \int dk_y
    \Big[
    f_{k_x,k_y,i}(t,x,y,z) \hat{a}_{k_x,k_y,i} + f^*_{k_x,k_y,i}(t,x,y,z) \hat{a}^\dagger_{k_x,k_y,i}
    \Big],
\end{align}
where $\hat{a}_{k_x,k_y,i}$ and $\hat{a}^\dagger_{k_x,k_y,i}$ are canonically commuting creation and annihilation operators for the modes $f_{k_x,k_y,i}$, and the vacuum state of the theory is defined as the state $\ket 0$ which is annihilated by all the annihilation operators. We will assume that the modes appearing in the sum are normalizable, so that the excitations created by the creation operators are normalizable excitations.\footnote{This implies that the states created by acting on the vacuum with the creation operators have position space wavefunctions which are $L^2$ normalizable.} Since we have chosen the modes $f$ to be positive frequency with respect to the timelike Killing vector field $\partial_t$, this canonical vacuum state is the lowest energy state with respect to the Hamiltonian generating time translations in $t$. 

In order for the field $\phi$ to canonically commute with its conjugate momentum field, it is necessary and sufficient that the mode functions $f_{k_x,k_y,i}$ are orthonormal in the Klein-Gordon inner product. This implies the normalization condition 
\begin{align}\label{eq:u_normalization}
    \int_{-z_0}^{z_0} dz \,u_{i}(z)u_{j}(z) = \delta_{ij}.
\end{align}
Since the $z$ coordinate has a finite range, the set of normalizable solutions $u_i$ to the Schr\"odinger equation is discrete, justifying the use of the discrete label $i=0,1,2,\dots$.

Having canonically quantized the scalar field, we can read off its energy spectrum. The Hamiltonian for a free scalar field is 
\begin{align}
    \hat{H} = \sum_i \int dk_x\int dk_y\, \omega_{k_x,k_y,i} \, \left(\hat{a}^\dagger_{k_x,k_y,i}\hat{a}_{k_x,k_y,i}+\frac{1}{2}\right),
\end{align}
and thus, from Eq. \eqref{eq:phi_energy_spectrum}, the energy spectrum of single particle excitations of the field $\phi$ is given by $\omega = \sqrt{k_x^2+k_y^2+\lambda_i}$. Since this is a free theory, the multi-particle spectrum is simply obtained by adding together single particle excitations. 

We expect that the spectrum of excitations in the wormhole is equal to the spectrum of excitations in the dual microscopic theory. The dispersion relation $\omega = \sqrt{k_x^2+k_y^2+\lambda_i}$ suggests that in the (2+1)-dimensional microscopic theory\footnote{Note that the dual microscopic theory is (3+1)-dimensional (two 3D CFTs coupled by 4D auxiliary degrees of freedom), as we have pointed out. However, it flows in the IR to a gapped effective (2+1)-dimensional theory \cite{Antonini2022}. When, here and in the rest of the paper, we refer to the full dual theory as being (2+1)-dimensional, we have in mind such an effective gapped theory. This should not be confused with the two 3D CFTs living at the two asymptotic boundaries, and coupled by the auxiliary 4D degrees of freedom.} the particle excitations have momenta $\vec k=(k_x,k_y)$ and a discrete spectrum of squared masses $m_i^2 = \lambda_i$. Indeed, in the particular construction considered in \cite{VanRaamsdonk:2021qgv,Antonini2022}, the dual microscopic theory flows in the infrared towards a confining gauge theory on $\mathbb R^{1,2}$, which we expect to exhibit a discrete particle spectrum. 

Therefore we find that the wormhole scale factor, and hence the corresponding FRW scale factor, determines the mass spectrum of scalar particles in the underlying microscopic theory. It would be interesting to understand which microscopic theories correspond to FRW cosmologies with realistic scale factors. We will not pursue this question here, but it could be an interesting area for future work. Instead we will now discuss correlators of scalar operators in the microscopic theory, and how they are related to the geometry of the bulk wormhole.

\subsection{Scalar two-point functions}\label{sec:two_point_functions}

We would like to understand, given the geometry of the wormhole, what are the correlation functions of the CFT operator dual to a bulk scalar field. Then in Section \ref{sec:wormhole_from_CFT} we will study the converse problem: how to reconstruct the wormhole geometry from a knowledge of the microscopic correlators. 

Computing CFT correlators from a knowledge of the bulk geometry is straightforwardly done using the extrapolate dictionary of AdS/CFT: we compute the bulk correlators, and extrapolate them to the boundary to obtain correlators of 3D CFT operators in the dual confining theory. Let us work this out for the two-point function.

From \eqref{eq:field_expansion}, the bulk vacuum Wightman function evaluates to\footnote{In the rest of the paper we will omit the $\hat{}$ on quantum operators.}
\begin{align}
    G^{+}(x,x')\equiv \bra{0} \phi(x)\phi(x')\ket{0}
    &=
    \sum_{n}\int d k_x\int d k_y f_{k_x,k_y,n}(x)f_{k_x,k_y,n}^*(x'). 
\end{align}
Using the expression \eqref{eq:f_mode_functions} for the mode functions, this can be simplified to 
\begin{align}
    G^{+}(x,x')
    =
    \frac{1}{4\pi a(z)a(z')}\sum_{n}
    u_n(z)u_n(z')
    \int_{k_n}^\infty d\omega e^{-i\omega\Delta t} J_0\left(\Delta x_\perp\sqrt{\omega^2-k_i^2}\right),
\end{align}
where $\Delta t \equiv t-t'$, $\Delta x_\perp \equiv \sqrt{(x-x')^2+(y-y')^2}$, and we recall that $u_n(z)$ are eigenfunctions of the Schr\"odinger equation \eqref{eq:SE}, with corresponding eigenvalues $k_n^2\equiv \lam_n$. Finally, the $\omega$ integral can be evaluated to give \cite{gradshteyn2014table}
\begin{align}
    G^{+}(x,x') 
    &=
    \frac{1}{4\pi a(z)a(z')\sqrt{-\Delta t^2+\Delta x_\perp^2}}
    \sum_{n=0}^\infty
    u_n(z)u_n(z')
    \exp\left(-k_n\sqrt{-\Delta t^2+\Delta x_\perp^2}\right),
\end{align}
with the square roots taking values on their principal branch.

Consider now the CFT two-point function of the operator $\mathcal O$ which is dual to the field $\phi$. This can be obtained by taking the points $z$ and $z'$ to the boundary, and removing powers of $z_0\pm z$ and $z_0\pm z'$ so that the expression is independent of $z$ and $z'$. Concretely, $a(z)\sim 1/(z_0-|z|)$ as $z\rightarrow \pm z_0$. Meanwhile the mode functions are of the form $u_n(z)\sim u_n^+(z_0+z)^{\Delta_+}$ near $z=-z_0$, where $u_n^+$ is the normalizable coefficient, chosen such that the modes are normalized as in \eqref{eq:u_normalization}. Since $u_n(z)$ is even if $n$ is even and odd if $n$ is odd, near $z=+z_0$ the mode functions are of the form $u_n(z)\sim (-1)^n u_n^+(z_0-z)^{\Delta_+}$.

Since the wormhole has two boundaries, at $z=\pm z_0$, we have two types of insertions of the operator $\mathcal O$ into correlation functions. We denote an insertion of $\mathcal O$ at position $(t,x,y)$ at the boundary $z=-z_0$ by $\mathcal O_-(t,x,y)$, and similarly for an insertion at the boundary $z=+z_0$. Thus there are two types of two-point functions: $\langle \mathcal O_\pm\mathcal O_\pm\rangle$ and $\langle \mathcal O_\pm\mathcal O_\mp \rangle$. The extrapolation procedure gives
\begin{align}\label{eq:O+O+}
    \langle \mathcal O_\pm(t,x,y)\mathcal O_\pm(t',x',y')\rangle 
    &=
    \frac{1}{4\pi\sqrt{-\Delta t^2+\Delta x_\perp^2}}
    \sum_{n=0}^\infty
    (u_n^+)^2
    \exp\left(-k_n\sqrt{-\Delta t^2+\Delta x_\perp^2}\right),
\end{align}
for the one sided correlators, and
\begin{align}\label{eq:O+O-}
    \langle \mathcal O_\pm(t,x,y)\mathcal O_\mp(t',x',y')\rangle 
    &=
    \frac{1}{4\pi\sqrt{-\Delta t^2+\Delta x_\perp^2}}
    \sum_{n=0}^\infty
    (-1)^n(u_n^+)^2
    \exp\left(-k_n\sqrt{-\Delta t^2+\Delta x_\perp^2}\right),
\end{align}
for the two sided correlators. Given the wormhole scale factor, the mode functions $u_n(z)$ and normalizations $u_n^+$ can be computed, and the $\mathcal O\mathcal O$ two-point functions can be evaluated. Therefore, the behavior of scalar operator correlators in the dual confining theory is uniquely fixed by the wormhole geometry. 

\subsection{Vector particle spectrum and massless sector from bulk gauge field}
\label{sec:gauge_field}

Let us now consider the quantization of a free $U(1)$ gauge field in the wormhole background (\ref{eq:wormhole}). We will report here our main results, while a complete, detailed analysis can be found in Appendix \ref{app:gauge_field}. For simplicity of notation, in this section we define $\ell=2z_0$ and shift $z\to z-z_0$; the two boundaries are therefore at $z=0,\ell$. The action for a massless gauge field is given by\footnote{In this section we will use latin indices $I,J=0,1,2,3$ to indicate 4D components, and greek indices $\mu,\nu=0,1,2$ for 3D components $t,x,y$.}
\begin{equation}
    S_{gauge}=-\frac{1}{4}\int d^4x\sqrt{-g}F_{IJ}F^{IJ}
    \label{eq:gaugeaction}
\end{equation}
leading to the equations of motion (which, together with the Bianchi identity for $F_{IJ}$, provide the four Maxwell's equations)
\begin{equation}
    \partial_I\left(\sqrt{-g}F^{IJ}\right)=0.
\end{equation}
As a result of the conformal invariance of the gauge field action (\ref{eq:gaugeaction}) in 4D, the equations of motion are independent of the wormhole scale factor $a(z)$. In particular, we will work in Lorenz gauge $\eta^{IJ} \partial_IA_J=0$, where they take the simple form
\begin{equation}
   \eta^{IJ}\partial_I\partial_JA_K=0.
   \label{eq:gaugeeom}
\end{equation}
We must now specify boundary conditions for the gauge field at the two boundaries. Unlike the scalar field case, the gauge field has two normalizable components, and the choice of boundary conditions leads to non-equivalent quantizations schemes (see Appendix \ref{app:gauge_field} and \cite{Aharony:2010ay,Hijano:2020szl,Witten:2003ya,Yee:2004ju,Marolf:2006nd}). We will focus here on two (customary) possible choices of boundary conditions\footnote{See \cite{Marolf:2006nd} for a detailed study of other possible choices of boundary conditions.}: Dirichlet (also known as ``magnetic''), where we fix $F_{\mu\nu}|_{z=0,\ell}$, and Neumann (also known as ``electric''), where we fix $F_{\mu z}|_{z=0,\ell}$. The first choice, which leads to the ``standard'' quantization of the gauge field, is equivalent (in Lorenz gauge) to fixing the leading term of $A_\mu$ and the subleading term of $A_z$ at the boundary, i.e. fixing $A_\mu|_{z=0,\ell}$ and $\partial_z A_z|_{z=0,\ell}$,  up to a residual gauge transformation; the second choice, which leads to the ``alternative'' quantization of the gauge field, is equivalent to fixing the subleading term of $A_\mu$ and the leading term of $A_z$ at the boundary, i.e. fixing $\partial_z A_\mu|_{z=0,\ell}$ and $A_z|_{z=0,\ell}$, up to a residual gauge transformation. The part of the bulk gauge field which is not fixed by the boundary conditions is quantized. Note that, motivated by the reflection symmetry of the microscopic theories of our interest \cite{Antonini2022}, we are choosing identical boundary conditions at the two asymptotic boundaries; in general, it is possible to consider different boundary conditions at the two boundaries. Let us now analyze the result of the quantization of the bulk gauge field in the two schemes.

\subsubsection*{Standard quantization}

Let us start with Dirichlet boundary conditions. After fixing the residual bulk gauge freedom, we can quantize the bulk gauge field, obtaining
\begin{equation}
    \begin{aligned}
        &\hat{A}_\mu^D=\sum_{n=1}^\infty \int \frac{dk_xdk_y}{2\pi\sqrt{\omega_n}\sqrt{\ell}}\sum_{j=a,b}\left[\tilde{\varepsilon}_\mu^{(j)}(k,\rho_n)\sin(\rho_n z)\textrm{e}^{-(i\omega_n t-k_xx-k_yy)}\hat{a}^{(j)}_{k,\rho_n}+h.c.\right]\\
        &\hat{A}_z^D=\int\frac{dk_xdk_y}{2\pi\sqrt{2\omega_0}\sqrt{\ell}}\left[\textrm{e}^{-(i\omega_0 t-k_xx-k_yy)}\hat{a}_{k}+h.c.\right]
    \end{aligned}
    \label{eq:dir_quant}
\end{equation}
where $\omega_n=\sqrt{k_x^2+k_y^2+\rho_n^2}$, $\rho_n=n\pi/\ell$, $\hat{a}$ and $\hat{a}^\dagger$ are creation and annihilation operators satisfying the canonical commutation relations, and we defined the polarization vectors
\begin{equation}
    \begin{aligned}
        &\tilde{\epsilon}_\mu^{(a)}(k,\rho_n)=\frac{1}{\sqrt{k_x^2+k_y^2}}(0,-k_y,k_x)\\
        &\tilde{\epsilon}_\mu^{(b)}(k,\rho_n)=\frac{1}{\rho_n\sqrt{k_x^2+k_y^2}}(-(k_x^2+k_y^2),\omega_n k_x,\omega_n k_y).
    \end{aligned}
    \label{eq:polarizations}
\end{equation}
The quantization of the 4D gauge field in the wormhole background with Dirichlet boundary conditions gives rise to a set of Kaluza-Klein modes. There is a zero mode, which can be interpreted as a massless 3D scalar field (the $z$ component of the gauge field, uniform in the $z$ direction), and an infinite tower of 3D massive vector fields, with discrete mass spectrum $m_n=\rho_n$. The $\tilde{\varepsilon}^{(a)}_\mu$ and $\tilde{\varepsilon}^{(b)}_\mu$ polarizations correspond to the transverse and longitudinal polarizations of such 3D massive vector fields, respectively. In the boundary dual confining theory, this bulk field content corresponds to a massless scalar particle and and a tower of massive vector bosons. 

We would like to point out two features emerging from the present analysis. First, unlike the spectrum of a scalar field analyzed in Section \ref{sec:particle_spectrum}, the spectrum of a 4D gauge field (or, equivalently, the mass spectrum of the dual vector particles) gives no information about the wormhole geometry besides the range of the $z$ coordinate.\footnote{The same would be true of a massless scalar with a conformal coupling to the Ricci scalar.} Since a free 4D gauge field is conformal, this feature is in fact to be expected. This means that the wormhole geometry cannot be reconstructed by starting from the mass spectrum of such vector fields.

Second, the presence of the massless scalar field indicates that, although the dual theory is confining, a massless sector is still present in a such theory. In order to understand how such a massless sector can arise, consider two copies of a 3D CFT with a global $U(1)$ symmetry group, whose respective bulk dual theories will then possess a $U(1)$ gauge symmetry. Let us now couple the two theories using 4D auxiliary degrees of freedom as described in \cite{VanRaamsdonk:2021qgv,Antonini2022}, such that the resulting theory is confining in the IR, and the holographic dual of its ground state is given by an eternal traversable wormhole of the form (\ref{eq:wormhole}). The bulk theory will now contain a single $U(1)$ gauge field, associated with the $U(1)\times U(1)$ global symmetry of the boundary theory. This is suggestive of the fact that the boundary global symmetry is spontaneously broken by the coupling: $U(1)\times U(1)\to U(1)$. The massless scalar field $\hat{A}_z$ arising in the present analysis can be interpreted as a Goldstone boson associated to this spontaneous symmetry breaking \cite{Antonini2022}. The presence of the massless sector guarantees the existence of correlations at arbitrarily long scales (in the non-compact directions) in the ground state of the confining theory, and therefore in the wormhole. This feature can help with solving the horizon problem in the FRW cosmology related to our wormhole by double analytic continuation \cite{Antonini2022,Antonini2022short}.

\subsubsection*{Alternative quantization}

With Neumann boundary conditions, after fixing the residual bulk gauge freedom, we obtain the bulk quantum field
\begin{equation}
\begin{aligned}
        &\hat{A}_\mu^N=\int \frac{dk_xdk_y}{2\pi\sqrt{2\omega_0}\sqrt{\ell}}\left[\tilde{\varepsilon}_\mu^{(0)}(k)\textrm{e}^{-(i\omega_0 t-k_xx-k_yy)}\hat{a}^{(0)}_{k}+h.c.\right]\\
        &+\sum_{n=1}^\infty \int \frac{dk_xdk_y}{2\pi\sqrt{\omega_n}\sqrt{\ell}}\sum_{j=a,b}\left[\tilde{\varepsilon}_\mu^{(j)}(k,\rho_n)\cos(\rho_n z)\textrm{e}^{-(i\omega_n t-k_xx-k_yy)}\hat{a}^{(j)}_{k,\rho_n}+h.c.\right]
\end{aligned}
\label{eq:neu_quant}
\end{equation}
with $\tilde{\varepsilon}^{(0)}_\mu (k)=(0,-k_y,k_x)/\sqrt{k_x^2+k_y^2}$ and $\tilde{\varepsilon}_\mu^{(a,b)}(k,\rho_n)$ as in equation (\ref{eq:polarizations}).
The quantization of the 4D gauge field in the wormhole background with Neumann boundary conditions also gives rise to a set of Kaluza-Klein modes. However, in this case there is a 3D massless gauge field with only one transverse physical polarization, and an infinite tower of 3D massive vector fields with mass $m_n=\rho_n$. No 3D massless scalar field is present. The 3D massless gauge field is uniform in the $z$ direction. In the boundary dual confining theory, this bulk field content corresponds to a boundary massless gauge field, and a tower of massive vector fields. The presence of a single bulk gauge field again suggests a $U(1)\times U(1)\to U(1)$ symmetry breaking pattern. The existence of a 3D massless gauge field and the absence of a Goldstone boson associated with the symmetry breaking suggests that the $U(1)$ symmetries of the two original 3D CFTs (and the remaining $U(1)$ after symmetry breaking) are now gauged. Once again, the presence of the massless sector associated with the 3D massless gauge field guarantees the existence of correlations at arbitrarily long scales in the non-compact directions. 

Therefore, we can conclude that the presence of a $U(1)$ gauge field in the wormhole geometry implies the existence of massive vector particles (with an evenly-spaced discrete spectrum of masses $m_n=n\pi/\ell$ ($n=1,2,3,...$)) and of a massless sector in the dual confining theory. For the standard quantization of the bulk gauge field, the massless sector is given by a massless scalar field, to be regarded as a Goldstone boson for the spontaneous breaking of a global $U(1)\times U(1)\to U(1)$ symmetry due to the coupling between the two 3D CFTs; for the alternative quantization, the spontaneously broken $U(1)\times U(1)\to U(1)$ symmetry is gauged, and the massless sector is given by the associated remaining boundary gauge field.

Additional particles in the dual confining gauge theory arise from fluctuations of the metric and other fields in the geometry; their mass spectra can be obtained via a similar analysis. 

\subsection{Entanglement entropy}
\label{sec:ententropy}

Another boundary observable in which we expect to see a clear signature of the presence of the wormhole is the entanglement entropy of subregions of the dual microscopic theory. For example, consider the behavior of the entanglement entropy for a region $A=A_L\cup A_R$ --- where $A_L$ and $A_R$ are identical-sized subregions of the left and right 3D CFTs respectively --- as we vary the size of $A$. In the presence of a wormhole, we expect to observe a phase transition in the entanglement entropy related to a bulk transition from a disconnected RT surface (dominant for small $A$) to a connected RT surface going through the wormhole (dominant for large $A$). The transition occurs because the disconnected surfaces have a regulated area that eventually grows like the volume of the boundary region, while the connected surface has an area that eventually grows like the area of the boundary region. In the absence of a wormhole geometry connecting the two 3D theories living on the two boundaries, there is no such phase transition. We leave further investigation of the properties of holographic entanglement entropies in our setup to future work.

\section{Reconstructing the wormhole}
\label{sec:wormhole_from_CFT}

So far we have discussed how properties of the microscopic theory, such as its particle spectrum and correlators, are related to the geometry of the wormhole. Now let us ask the converse question: can we reconstruct the wormhole geometry from properties of the underlying confined microscopic theory? We will find that such a reconstruction is indeed possible if we have access to certain observables of the microscopic theory.

The first example of such an observable is given by the two-point functions $\langle \mathcal O_1\mathcal O_1\rangle$ and $\langle \mathcal O_2\mathcal O_2\rangle$ in the confining theory of two 3D CFT scalar operators $\mathcal O_1$ and $\mathcal O_2$ of different scaling dimensions\footnote{The operators $\mathcal{O}_1,\mathcal{O}_2$ are operators associated with the two 3D CFTs living at the two asymptotic boundaries, but their expectation value is computed in the full confining theory obtained by coupling the two 3D CFTs by auxiliary 4D degrees of freedom.}. Our wormhole reconstruction algorithm from these two-point functions can be summarized as follows.

First, in Section \ref{sec:spectrum_from_correlator}, we show how a knowledge of the scalar two-point function $\langle \mathcal O \mathcal O\rangle$ in the microscopic confining theory allows one to obtain the mass spectrum of the associated scalar particle excitations, which is equal to the spectrum of modes of a bulk scalar field with some mass $m$. Then, in Section \ref{sec:inverse_SL_problem}, we show how the mass spectrum associated with an operator $\mathcal O_i$ can be used to reconstruct the Schr\"odinger potential which gives rise to this spectrum. From \eqref{eq:V(z)} we know that this potential is given directly in terms of the wormhole scale factor and the mass of the scalar field $\phi_i$ dual to the operator $\mathcal O_i$. In Section \ref{sec:scale_factor_from_potential} we show how a knowledge of two such potentials, associated with two scalar operators $\mathcal O_1$ and $\mathcal O_2$ of different dimensions, allows one to invert this relationship and obtain the scale factor.\footnote{Note that a single potential associated with a massless bulk scalar field is sufficient to reconstruct the scale factor. For $m^2<0$ (but above the BF bound) a single potential may also be enough to reconstruct the geometry, though we only have numerical evidence for this. For $m^2>0$, a single potential is not sufficient.} We summarize the entire algorithm in Section \ref{sec:reconstruction_algorithm_summary} and provide an example of reconstruction in Section \ref{sec:reconstruction_example}. The reconstruction algorithm for the wormhole geometry from boundary scalar two-point functions should be regarded as one of the main results of this paper. 
Finally, in Section \ref{sec:metric_from_heavy} and \ref{sec:metric_from_entanglement} we explain how the wormhole scale factor can be reconstructed from correlators of 3D CFT operators with large scaling dimension, and from the entanglement entropy of strip-shaped subregions of one of the two 3D CFTs.

\subsection{Mass spectrum from two-point function}
\label{sec:spectrum_from_correlator}

In equation \eqref{eq:O+O+} and \eqref{eq:O+O-} we wrote down the CFT two-point functions $\langle \mathcal O_\pm\mathcal O_\pm\rangle$, in which both insertions of the operator $\mathcal O$ are inserted into the same same wormhole boundary, and $\langle \mathcal O_\pm\mathcal O_\mp\rangle$, where one insertion is into the CFT at $z=+z_0$ and the other insertion into the CFT at $z=-z_0$. Namely
\begin{align}
    G_{++}(s)&\equiv \langle \mathcal O_\pm(x)\mathcal O_\pm(x')\rangle 
    =
    \frac{1}{4\pi s}
    \sum_{n=0}^\infty
    (u_n^+)^2
    \exp\left(-k_n s\right),\\
     G_{+-}(s)&\equiv \langle \mathcal O_\pm(x)\mathcal O_\mp(x')\rangle 
    =
    \frac{1}{4\pi s}
    \sum_{n=0}^\infty
    (-1)^n(u_n^+)^2
    \exp\left(-k_n s\right),
\end{align}
where $s\equiv \sqrt{-(t-t')^2+(\bm x-\bm x')^2}$.

Let us consider these correlators in the regime where $s$ is imaginary, $s=i\xi$, namely
\begin{align}
    G_{++}(\xi)
    &=
    \frac{1}{4\pi i \xi}
    \sum_{n=0}^\infty
    (u_n^+)^2
    \exp\left(-i k_n \xi\right),
    \\
    G_{--}(\xi)
    &=
    \frac{1}{4\pi i\xi}
    \sum_{n=0}^\infty
    (-1)^n(u_n^+)^2
    \exp\left(-i k_n \xi\right).
\end{align}
Multiplying by $i \xi$ and taking the Fourier transform gives
\begin{align}
    G_{++}(k)
    \equiv
    \int_{-\infty}^\infty d\xi \,i\xi G_{++}(\xi) e^{i k \xi} 
    &=
    \frac{1}{2}
    \sum_{n=0}^\infty
    (u_n^+)^2
    \delta(k-k_n),
    \\
    G_{+-}(k)
    \equiv
    \int_{-\infty}^\infty d\xi \,i\xi G_{+-}(\xi) e^{i k \xi} 
    &=
    \frac{1}{2}
    \sum_{n=0}^\infty
    (-1)^n(u_n^+)^2
    \delta(k-k_n),
\end{align}
and so we see that the spectrum of scalar excitations, $k_n^2 = \lam_i$, is simply obtained from the peaks in the Fourier transforms of the two-point correlators, while the normalizations $u_n^+$ are obtained from the amplitudes of the peaks. Interestingly, it is possible to obtain this information either from the one-sided or two-sided correlators. Notice that to obtain the spectrum above a given scale $k_0$, we need to probe the two-point function on distance scales $\xi<1/k_0$. From the analysis of Section \ref{sec:particle_spectrum}, we can identify this spectrum of scalar excitations with the discrete spectrum of modes associated with a bulk scalar field.

\subsection{The inverse Sturm-Liouville problem}\label{sec:inverse_SL_problem}

Suppose now that, perhaps starting from CFT correlators as in the previous section, we obtain the spectrum $\lambda_i$ of scalar excitations associated with a given operator $\mathcal O$ in the microscopic theory, which we can identify with the discrete spectrum of modes for a bulk scalar field. We know from the analysis in Section \ref{sec:particle_spectrum} that the $\lambda_i$ are eigenvalues to the Schr\"odinger equation \eqref{eq:SE},
\begin{align}\label{eq:SE2}
    -u_i''(z)+ V(z)u_i(z)= \lambda_i u_i(z),
\end{align}
where the potential $V(z)$ is related to the scale factor of the wormhole. Given the potential, the problem of finding the eigenvalues $\lambda_i$ (i.e. diagonalizing the Hamiltonian) is called a Sturm-Liouville problem. Here, we are interested in the inverse Sturm-Liouville problem problem: finding the potential $V(z)$ from a given spectrum $\lam_i$.\footnote{The question ``can one hear the shape of a drum?" \cite{kac1966} --- i.e. ``Can one determine the geometry of a two-dimensional manifold from a knowledge of the spectrum of the Laplace operator?" --- is a higher dimensional analog to the inverse Sturm-Liouville problem.}

Under the assumption of a generic potential the solution to the inverse Sturm-Liouville problem is not unique \cite{Levitan-Gasymov,Hald}. However the potentials that are of interest to us are not completely generic; we are assuming that the wormhole geometry is even $a(z)=a(-z)$, and so the potential is also even, $V(z)=V(-z)$. In the context of wormholes that arise from cosmological models \cite{Antonini2022}, this symmetry of the geometry arises from the $\mathbb Z_2$ symmetry present in both the Euclidean CFT path integral used to define the theory, as well as the choice of slicing of this path integral which defines the state of the theory. It has been shown in \cite{Hald} that if the potential is even and integrable, $\int_{-z_0}^{z_0}V(z)<\infty$, then the inverse Sturm-Liouville problem can be solved. Namely, there exists an efficient algorithm which allows one to reconstruct the potential $V(z)$ given its spectrum.

Unfortunately we cannot directly apply these results to solve our inverse Sturm-Liouville problem, because our potential diverges as $1/z^2$ near the AdS boundaries (see Eq. \eqref{eq:V_asymptotics}) and hence is not integrable. In this section we will upgrade the reconstruction algorithm in \cite{Hald} to allow for potentials with AdS asymptotics. We will make the following assumptions:

\subsubsection*{Assumptions}\label{sec:assumptions}
\begin{enumerate}
    \item The domain of $z$ is $(-z_0,z_0)$ for some $z_0>0$;
    \item The potential is symmetric, $V(z)=V(-z)$;
    \item The potential diverges as $\frac{\alpha}{(z_0- |z|)^2}+\frac{\beta}{(z_0-|z|)}$ at $z=\pm z_0$, with $\alpha>-1/4$ and $\beta\in \mathbb R$;
    \item The normalizable spectrum $\{\lambda_j\}$ to the Schr\"odinger equation \eqref{eq:SE2} is known.
\end{enumerate}
Let us make some comments about these assumptions. From equation \eqref{eq:V_asymptotics} we see that the parameter $\alpha$ characterizing the leading divergence of the potential is given by $\alpha = 2+m^2 L^2$, where $L$ is the AdS length in the asymptotically AdS regions, and $m$ is the mass of the scalar field. The assumption $\alpha>-1/4$ is therefore equivalent to the Breitenlohner-Freedman (BF) bound $m^2L^2>-d^2/4$, with $d=3$ the spatial dimension of the wormhole \cite{breitenlohner1982stability,breitenlohner1982positive}. We are also allowing for a subleading divergent term in $V(z)$ proportional to some other constant $\beta$.\footnote{In principle, we could also have some subleading divergent term with a different non-integer power, but we will restrict our analysis to this case.} We see from \eqref{eq:V(z)} that such a term in the potential is expected to arise from general scale factors with boundary asymptotics of the AdS form. 

As we have already seen, for a potential with leading divergence of the form in assumption 3, the series expansion near $z=-z_0$ of a solution to the Schr\"odinger equation \eqref{eq:SE2} is
\begin{align}\label{eq:u_asymptotics}
    &u(z)=
    \um(-z_0)(z+z_0)^{\Delta_-}
    \left(1+\dots\right)
    +
    \up(-z_0)(z+z_0)^{\Delta_+}
    \left(1+\dots\right),
\end{align} 
where
\begin{align}
    &\Delta_{\pm} \equiv \frac{1}{2}\pm\Dbar,\label{newdeltas}\\
    &\Dbar \equiv \frac{1}{2}\sqrt{1+4\alpha},
\end{align}
and $u^{\pm}(-z_0)$ are numerical coefficients. The omitted terms in the above expansion are of subleading order in the distance $z+z_0$ from the boundary. We say that the term proportional to $\up(-z_0)$ is normalizable at $z=-z_0$, and the term proportional to $\um(-z_0)$ is non-normalizable at $z=-z_0$. We have a similar expansion at $z=+z_0$, with normalizable and non-normalizable coefficients $\up(z_0)$ and $\um(z_0)$. We define the \textit{normalizable eigenfunctions} of the Schr\"odinger equation \eqref{eq:SE2} to be those eigenfunctions for which the non-normalizable component vanishes at both endpoints, $\um(-z_0)=\um(z_0)=0$. The \textit{normalizable spectrum} is the set of eigenvalues associated with the normalizable eigenfunctions.

Note that the normalizable eigenfunctions are indeed normalizable in $L^2$ norm. Perhaps somewhat confusingly, if $\alpha<3/4$, eigenfunctions which are non-normalizable at one or both endpoints are also normalizable in $L^2$ norm. This is related to the ambiguity --- which arises when a scalar field in AdS has a mass in the range $-d^2/4<m^2L^2<-d^2/4+1$ --- in deciding which part of the field to identify with the expectation value of the dual CFT operator, and which part to identify with the source of the dual CFT operator. For simplicity we will always associate the normalizable eigenfunctions with the expectation values of the dual operators, which in the bulk is the statement that, as we have done in Section \ref{sec:particle_spectrum}, we will expand the scalar field in terms of a complete set of normalizable eigenfunctions of the wave equation. 

We will now show how to reconstruct the potential $V(z)$ from its spectrum, subject to these assumptions. The main idea is to reconstruct the unknown potential $V(z)$ from some known ``test'' potential $\Vt(z)$ which has the same asymptotic spectrum as $V(z)$. Our first task is to determine the test potential $\Vt(z)$.

\subsubsection*{Finding the test potential}

We want to find a test potential $\Vt(z)$ such that the eigenvalues $\lamt_j$, $j=0,1,2,\dots$, of the eigenvalue problem
\begin{align}\label{eq:Vt_evalue_problem}
    -u''(z)+\Vt(z)u(z)&=\lamt u(z),\\
    \um(\pm z_0)&=0\label{boundarycond},
\end{align}
are asymptotically equal to $\lam_j$ for $j\rightarrow\infty$. To find such a $\Vt(z)$, let us consider solving the equation $-u''(z)+V(z)u(z) = \lam u(z)$ on the interval $z\le 0$, in the limit $\lam\rightarrow\infty$. If $V(z)$ contains terms which are bounded on $z\le 0$, then the term proportional to $\lam$ dominates over these terms and we can simply replace these bounded terms by their average value over $z\le 0$. Besides this constant, the only terms which remain are the divergent pieces of $V(z)$ on $z\le 0$, which we know by assumption to be $\frac{\alpha}{(z_0+z)^2}+\frac{\beta}{z_0+z}$, although we do not know the values of $\alpha$ and $\beta$ a priori. By making the same argument for $z> 0$, we conclude that in the limit $\lam\rightarrow \infty$, the potential $V(z)$ can effectively be replaced by 
\begin{align}
    \Vt(z) \equiv
    \frac{\alpha}{(z_0-|z|)^2}+\frac{\beta}{z_0-|z|}+c,
    \label{eq:Vt}
\end{align}
where $\alpha,\beta,c$ are unknown constants. Our goal is to determine these constants. We will do this by equating the asymptotic spectrum $\lamt_j$ associated with this $\Vt(z)$ with the asymptotic part of the given spectrum $\lam_j$.

The potential $\Vt(z)$ is simple enough that we can explicitly compute its asymptotic spectrum. As we show in Appendix \ref{app:asymp_spectrum}, up to terms that go to zero as $j\rightarrow\infty$, the spectrum at large $j$ is
\begin{align}
    \tilde{\lam}_j\sim 
    Z j^2+
    A j+
    B \log(j)+
    C,
    \label{eq:fitformula}
\end{align}
where $Z,A,B,C$ are given by
\begin{align}
    Z &= \left(\frac{\pi}{2z_0}\right)^2,\label{eq:Z}\\
    A &= \left(\frac{\pi}{2z_0}\right)^2 2\Delta_+,\\
    B &= \frac{\beta}{z_0},\\
    C &= \left(\frac{\pi}{2z_0}\right)^2\left(\Delta_+^2-\frac{4\alpha}{\pi^2}\right)
    +\frac{\beta}{z_0}\left[\log\left(\pi\right)-\psi(\Delta_+)\right]+c,\label{eq:C}
\end{align}
with $\psi(x)$ being the digamma function and $\Delta_+$ was defined in \eqref{newdeltas}. Notice that $Z>0$. As we saw, the Breitenlohner-Freedman bound requires $\alpha>-1/4$, and so $A>0$ as well. Meanwhile, $B$ and $C$ can take on any real values. These restrictions on $Z,A,B,C$, together with equation \eqref{eq:fitformula}, constitute the necessary asymptotic conditions on the sequence $\lam_j$ in order for it to correspond to the spectrum of a scalar field in a symmetric wormhole geometry which connects two asymptotically AdS regions. We will assume that the spectrum from which we are trying to reconstruct the wormhole geometry is of this form. In this case we know that there exists a solution to the inverse Sturm-Liouville problem; we will try to find this solution.

Recall that our current goal is to obtain the test potential $\Vt(z)$ defined in \eqref{eq:Vt}. Given the spectrum $\lam_j$ we can deduce the values $Z,A,B,C$, and then we can invert equations \eqref{eq:Z}-\eqref{eq:C} to give $z_0,\alpha,\beta$ and $c$ in terms of $Z,A,B,C$:
\begin{align}
    z_0 &= \frac{\pi}{2\sqrt Z}, \label{eq:z_0_vs_Z}\\
    \alpha &= \frac{A}{4Z}\left(\frac{A}{Z}-2\right),\label{eq:a_vs_A}\\
    \beta &= \frac{\pi B}{2\sqrt Z},\label{eq:b_vs_B}\\
    c &= -\frac{A}{4\pi^2}\left[\frac{(\pi^2-4)A}{Z}+8\right]
    -B\left[\log\left(\pi\right)-\psi\left(\frac{A}{2Z}\right)\right]+
    C.\label{eq:c_vs_C}
\end{align}
Therefore, given a spectrum $\lam_j$ with the correct asymptotic form, we can obtain the test potential $\Vt(z)$. The spectrum $\lamt_j$ associated with the test potential will be asymptotically equal to $\lam_j$, as desired.

Notice that, solely from this asymptotic analysis, we can already make several conclusions about some large scale features of the wormhole geometry, and the associated Schr\"odinger potential $V(z)$. For instance, from \eqref{eq:Z} we see that the leading (quadratic in $j$) term in the asymptotic expansion of the spectrum $\lam_j$ determines $z_0$, the coordinate width of the wormhole geometry. Meanwhile the term linear in $j$ determines the leading ($\alpha/z^2$) divergence of the potential near the AdS boundaries, and hence the cosmological constant. The subleading $\log(j)$ term gives the $\beta/z$ divergence of the potential, and the term constant in $j$ gives the average value $c$ of the non-singular terms in the potential. 

It is in accordance with our intuition from AdS/CFT that the values $\alpha$ and $\beta$, which are related to the AdS asymptotics of the wormhole geometry, are associated with the UV asymptotics of the spectrum $\lam_j$. Interestingly we also see that what might be considered the two most extremely IR features of the wormhole --- the size of the wormhole $z_0$ and the average value $c$ --- are also determined by the asymptotic form of the spectrum. Heuristically we can say that the asymptotic spectrum encodes the large scale features of the wormhole.

In order to determine the finer details of the wormhole, and in particular the potential $V(z)$, we will need to make use of our knowledge of the spectrum $\lam_k$ at finite $k$. In fact we will find that the eigenvalue $\lambda_k$ encodes the features of the potential on coordinate distance scales of the order $z_0/k$, and so if we are only interested in reconstructing the features of the potential (and thus of the scale factor) on scales larger than $z_0/N$, then to a good approximation it is enough to only know the finite subset $\{\lam_1,\lam_2,\dots,\lam_N\}$ of the full spectrum.\footnote{Of course we also need the asymptotic form of $\lam_j$ in order to obtain the large scale structure of the potential.} Let us now work out the details of this reconstruction.

\subsubsection*{Reconstructing the true potential from the test potential}

Having obtained the test potential $\Vt(z)$ from the asymptotic spectrum, let us now discuss how to reconstruct the true potential $V(z)$. The main idea of the reconstruction algorithm is based on the study of inverse Sturm-Liouville problems by Ref. \cite{Hald}, but adapted to allow for singular potentials such as ours. 

We begin by defining $u(z,\lam), \ut(z,\lam),v(z,\lam)$ and $\vt(z,\lam)$ as the solutions to the following initial value problems on $z\in(-z_0,z_0)$:  
\begin{align} \label{eq:u_and_v}
u:
    \begin{cases}
    u''(z,\lam)=\big[V(z)-\lam\big]u(z,\lam)\\
    \up(-z_0,\lam)=\frac{1}{2\Dbar}\\
    \um(-z_0,\lam)=0
    \end{cases}\quad
v:
    \begin{cases}
    v''(z,\lam)=\big[V(z)-\lam\big]v(z,\lam)\\
    v^{(+)}(+z_0,\lam)=\frac{1}{2\Dbar}\\
    v^{(-)}(+z_0,\lam)=0
    \end{cases}\\
\ut:\label{eq:ut_and_vt}
    \begin{cases}
    \ut''(z,\lam)=\big[\Vt(z)-\lam\big]\ut(z,\lam)\\
    \utp(-z_0,\lam)=\frac{1}{2\Dbar}\\
    \utm(-z_0,\lam)=0
    \end{cases}\quad
\vt:
    \begin{cases}
    \vt''(z,\lam)=\big[\Vt(z)-\lam\big]\vt(z,\lam)\\
    \vt^{(+)}(+z_0,\lam)=\frac{1}{2\Dbar}\\
    \vt^{(-)}(+z_0,\lam)=0
    \end{cases}
\end{align}
Here primes denote $z$ derivatives and recall that $\Dbar = \sqrt{\alpha+1/4}$. As before we are using the notation $u^{(\pm)}(z_0,\lam)$ for the normalizable ($+$) and non-normalizable ($-$) parts of $u(z,\lam)$ at $z=z_0$, and similarly at $z=-z_0$. The reason for the choice of normalization $1/(2\Dbar)$ is explained in Appendix \ref{app:lemma}. Notice that while $u$ and $\ut$ are by definition always normalizable (and hence non-diverging) at the left boundary $z=-z_0$, in general they are non-normalizable (and thus in general diverging) at the right boundary $z=+z_0$. Similarly the $v$ and $\vt$ are normalizable at the right boundary, but non-normalizable at the left one. However the discrete set of functions $u_j(z)\equiv u(z,\lam_j)$ and $v_j(z)\equiv v(z,\lam_j)$ are eigenfunctions of the Schr\"odinger equation with potential $V$ and normalizable boundary conditions at the left \textit{and} right boundaries, and so these functions are completely regular on the entire wormhole. Analogously, $\tilde{u}_j(z)\equiv \tilde{u}(z,\tilde{\lambda}_j)$, $\tilde{v}_j(z)\equiv \tilde{v}(z,\tilde{\lambda}_j)$ are also regular in $[-z_0,z_0]$.

With this in mind, we define the \textit{Wronskian} $\omega(\lam)$ as the non-normalizable part of $u(z,\lam)$ at the right boundary, namely
\begin{align}\label{eq:omega}
    \omega(\lambda)\equiv-\lim_{z\rightarrow z_0}\frac{u(z,\lam)}{(z_0-z)^{\Delta_-}},
\end{align}
so that $\omega(\lam)=0$ if and only if $\lam$ is a normalizable eigenvalue of the Schr\"odinger equation \eqref{eq:SE2} with potential $V(z)$. The following lemma, which we prove in Appendix \ref{app:lemma}, will play a crucial role in our reconstruction of $V(z)$ from $\Vt(z)$:

\begin{lemma} 
For any $L^2$ integrable function $f(z)$ on $z\in(-z_0,z_0)$ the following identity holds:
\begin{align}\label{eq:modified_SL_expansion}
    f(z)
    =
    \sum_{j=0}^\infty
    \frac{\vt_j(z)\int_{-z_0}^z dy\,u_j(y) f(y)+\ut_j(z)\int_{z}^{z_0} dy\,v_j(y) f(y)}{\omega'(\lam_j)}.
\end{align}
\end{lemma}

Let us set $f(z)$ equal to $u_0(z)$, the eigenfunction of the Schr\"odinger equation corresponding to the lowest eigenvalue $\lam_0$. A basic fact from Sturm-Liouville theory states that $u_j(z)$ has $j$ roots, and so $u_0(z)>0$. Using the orthonormality of the $u_j(z)$ in the $L^2$-norm, the lemma gives
\begin{equation}\label{eq:u0_expansion}
    u_0(z)=\ut_0(z)+\frac{1}{2}\sum_j \yt_j\int_{-z_0}^{z_0}dx\,u_j(x)u_0(x),
\end{equation}
where
\begin{equation}
    \yt_j(z)
    \equiv 
    2\frac{\vt_j(z)-(-1)^j \ut_j(z)}{\omega'(\lam_j)},
    \label{eq:yt}
\end{equation}
Differentiating \eqref{eq:u0_expansion} twice and using the differential equations satisfied by $u_0(z)$ and $\ut_0(z)$, we obtain
\begin{equation}
    V(z)u_0(z)=\Vt(z)u_0(z)+\sum_j\big[\yt_j(z)u_j(z)\big]'u_0(z).
\end{equation}
Finally, since $u_0(z)>0$ this implies
\begin{align}\label{eq:recon}
    V(z)=\Vt(z)+\sum_j\big[\yt_j(z)u_j(z)\big]'.
\end{align}

Equation \eqref{eq:recon} relates the unknown, true potential $V(z)$ to the test potential $\Vt(z)$, the latter having already been determined in the previous section. However this formula is not sufficient for determining $V(z)$, because the eigenfunctions $u_j(z)$ appearing inside this formula implicitly depend on $V(z)$ through the Schr\"odinger equations
\begin{align}\label{eq:SE3}
    u_j''(z)=\big[V(z)-\lambda_j\big]u_j(z).
\end{align}
Instead, we should think of \eqref{eq:recon} together with the Schr\"odinger equations \eqref{eq:SE3} as a coupled system of non-linear ODEs for the unknown functions $V,u_0,u_1,u_2,\dots$. 

If we restrict $j$ to the interval $0,1,\dots, \jmax$, then
\begin{align}\label{eq:system}
\begin{cases}
    V(z)=\Vt(z)+\sum_{j=0}^{\jmax}\big[\yt_j(z)u_j(z)\big]',\\
    u_j''(z)=\big[V(z)-\lambda_j\big]u_j(z),\quad j = 0,1,\dots,\jmax
\end{cases}
\end{align}
forms a closed system of $\jmax+2$ equations for $\jmax+2$ unknown functions $V,u_0,u_1,\dots,u_{\jmax}$. For a given value of $\jmax$, denote $V_{\jmax}(z)$ as the solution for $V(z)$ obtained by solving this finite system of equations. Since $\yt_j\rightarrow 0$ as $j\rightarrow\infty$\footnote{To see this note that as $j\rightarrow\infty$ we have that $\lam_j\rightarrow\lamt_j$ and hence in this limit $\ut_j$ and $\vt_j$ both approach the $j$th eigenfunction of the Schr\"odinger equation with potential $\Vt$. Since the $j$th eigenfunction of a symmetric Sturm-Liouville problem has parity $(-1)^j$ we see that $\vt_j=(-1)^j \ut_j$ and hence $\yt_j\rightarrow 0$.}, we expect that the corrections to $V_{\jmax}$ become smaller and smaller as we increase $\jmax$ to larger and larger values, and that in the limit $\jmax\rightarrow\infty$ we obtain the true potential $V(z)$. 

In practice, we will restrict to a finite $\jmax$ and thus obtain an approximation to the true potential $V(z)$. Since the functions $u_j$ and $\yt_j$ oscillate on scales of order $z_0/j$, we can think of this as being a good approximation to the true potential on coordinate distance scales larger than $z_0/j_{max}\sim 1/(\sqrt{Z}\jmax)$.\footnote{The \textit{scale independent} features of the potential $V(z)$, namely the amplitudes of the $1/z^2$ and $1/z$ AdS divergences at the AdS boundaries and the average value of the non-diverging piece of $V(z)$, were determined from the asymptotic (large $j$) part of the spectrum.} We will explicitly see this feature in the example reconstruction presented in Section \ref{sec:reconstruction_example}.

Let us briefly comment on a subtlety which arises when we attempt to solve the truncated system of ODEs \eqref{eq:system}. The issue is that in order to obtain $\yt_j$ via equation \eqref{eq:yt}, we need to compute the derivatives of the Wronskian, $\omega'(\lam_j)$. However the Wronskian is defined in \eqref{eq:omega} in terms of the mode function $u(z,\lam)$, which depends on the potential $V(z)$, and so it appears that $\omega'(\lam_j)$ must be solved for simultaneously with the system of equations for $V(z)$ and $u_j$. Although this is possible, there is an elegant approach due to \cite{Hald} which allows us to independently and efficiently compute $\omega(\lam_j)$ prior to computing $V(z)$ and $u_j$.

The main idea is to observe that, from its definition, the function $\omega(\lam)$ is equal to zero if and only if $\lam = \lam_j$. It can then be shown \cite{Hald} that $\omega(\lam)$ is an entire function of order $1/2$ and thus by the Hadamard factorization theorem can be written as $\omega(\lam) = A\prod_i(1-\lam/\lam_i)$, where $A$ is some normalization constant which can be determined by considering the large $\lam$ limit, in which $\lam_j\rightarrow \lamt_j$ and $u_j \rightarrow \ut_j$. Differentiating and setting $\lam = \lam_j$ gives 
\begin{align}\label{eq:omega_prime_approx}
    \omega'(\lam_j)
    =
    -\frac{\ut_j^{(-)}(z_0)}{\lam_j-\lamt_j}\prod_{i\neq j}^{\infty}\frac{(\lam_j-\lam_i)}{(\lam_j-\lamt_i)}
    \approx
    -\frac{\ut_j^{(-)}(z_0)}{\lam_j-\lamt_j}\prod_{i\neq j}^{\jmax}\frac{(\lam_j-\lam_i)}{(\lam_j-\lamt_i)},
\end{align}
where we have approximated $\lam_j\approx \lamt_j$ for $j>\jmax$. This formula allows us to approximate $\omega'(\lam_j)$ from quantities that we are given ($\lam_j$) and quantities which we can easily compute ($\lamt_j$ and $\ut_j$).

\subsection{Solving for the scale factor from $V(z)$}\label{sec:scale_factor_from_potential}

Having reconstructed the potential $V$, we now want to obtain the wormhole scale factor by solving the non linear ODE (\ref{eq:V(z)}) for $a(z)$. It is convenient to define the scale factor in AdS units $\tilde{a}(z)=a(z)/L$, in terms of which the ODE takes the form
\begin{equation}
  \tilde{a}''(z)+(\alpha-2)\tilde{a}^3(z)-V(z)\tilde{a}(z)=0
    \label{ODE}
\end{equation}
where we used $m^2L^2=\alpha-2$, and $\alpha$ is the parameter entering the definition of the test potential. In order to find a unique solution for $\tilde{a}(z)$, we need to impose initial conditions for $\tilde{a}(z)$ and its derivative at some point $z\in (-z_0,z_0)$. However, we only know that the scale factor is symmetric (and therefore $\tilde{a}'(0)=0$), and that its asymptotic behavior is given by equation (\ref{eq:a_asymp_ads}), which in AdS units becomes
\begin{equation}
    \lim_{z\to\mp z_0}\tilde{a}(z)=\frac{1}{z_0\pm z}.
    \label{limit}
\end{equation}
In general settings and without further assumptions, this information is not enough to uniquely determine the scale factor from a single reconstructed potential. Let us understand this in some more detail.

If the bulk scalar field is massless, which implies $\alpha = 2$, there is a unique symmetric solution to the ODE (\ref{ODE}) with the correct asymptotics (\ref{limit}). In fact, in this case the ODE reduces to a linear ODE. Therefore, given a symmetric solution of the ODE with some initial condition $\tilde{a}(0)$, all other solutions differ by a multiplicative constant. Only one solution in such family has the correct coefficient (i.e. $1$) for the divergent term. We conclude that a single potential associated to a bulk massless scalar field is enough to reconstruct the wormhole geometry. On the other hand, for $\alpha\neq 2$ (corresponding to $m^2\neq 0$) it is non-trivial to understand whether there is a unique symmetric solution to (\ref{ODE}) with the correct asymptotics.

In order to shed light on this problem, it is useful to consider the generic ansatz\footnote{Note that this is the most generic ansatz giving rise to a potential of the form we are considering (specified in Assumption 3 in Section \ref{sec:inverse_SL_problem}). In particular, it implements the symmetry of the scale factor and the AdS asymptotic behavior at the two boundaries, while introducing no further divergences in the potential besides the quadratic and linear ones appearing in our assumptions.}
\begin{equation}
    \tilde{a}(z)=\frac{1}{z_0+z}+\frac{1}{z_0-z}+R(z)
    \label{ansatz}
\end{equation}
where $R(z)$ is a symmetric function, regular in $z\in [-z_0,z_0]$. Suppose we have a solution $\tilde{a}_1(z)$ to the ODE (\ref{ODE}) of the form (\ref{ansatz}). Let us consider a second solution of the form $\tilde{a}_2(z)=\tilde{a}_1(z)+f(z)$ where $f(z)$ is a symmetric function, regular in $z\in [-z_0,z_0]$. $\tilde{a}_2(z)$ is by construction symmetric and of the form (\ref{ansatz}), i.e. it has the correct asymptotic behavior. Substituting $\tilde{a}_2(z)$ in the ODE (\ref{ODE}) and using the fact that $\tilde{a}_1(z)$ is also a solution, we find that $f(z)$ must satisfy
\begin{equation}
    \begin{cases}
    \frac{f''(z)}{f(z)}-\frac{\tilde{a}_1''(z)}{\tilde{a}_1(z)}+(\alpha-2)[f(z)+2\tilde{a}_1(z)][f(z)+\tilde{a}_1(z)]=0\\
    f'(0)=0\\
    f(\pm z_0)<\infty
    \end{cases}
    \label{addsol}
\end{equation}
For $\alpha=2$ (massless bulk scalar field), the problem (\ref{addsol}) has no non-trivial solutions, because all the candidate solutions of (\ref{addsol}) are proportional to $\tilde{a}_1(z)$, which is divergent at $z=\pm z_0$, and therefore violate $f(\pm z_0)<\infty$. This result is supported by our numerical analysis, and confirms our expectation that there is a unique symmetric solution of the ODE (\ref{ODE}) with asymptotic behavior (\ref{limit}) when $\alpha=2$: the spectrum of a massless scalar field is sufficient to reconstruct the wormhole geometry.

For $\alpha >2$ (massive, non-tachyonic bulk scalar field, $m^2>0$), our numerical analysis suggests that, in general, there are infinitely many solutions to the system \eqref{addsol}. For example, consider the scale factor\footnote{Note that this form of the scale factor has no particular physical meaning. We included the last term to emphasize that the presence of non-perturbative contributions (at $z=\pm z_0$) does not prevent the existence of multiple solutions. The same result was obtained for several different forms of the scale factor.}
\begin{equation}
    \tilde{a}(z)=\frac{1}{z_0+z}+\frac{1}{z_0-z}+\cos\left(\frac{2\pi z}{z_0}\right)+\exp\left[-\frac{1}{(z_0^2-z^2)}\right],
    \label{testscale}
\end{equation}
from which we can easily compute the corresponding potential $V(z)$. However, the system \eqref{addsol} then has multiple solutions, only one of which corresponds to the scale factor $\tilde a(z)$ which we started with (see Figure \ref{fig:multsol}). Without additional information it is not possible to pick out the correct scale factor over the other valid solutions, $\tilde a_1, \tilde a_2$, etc. 
\begin{figure}
    \centering
    \includegraphics[width=0.7\textwidth]{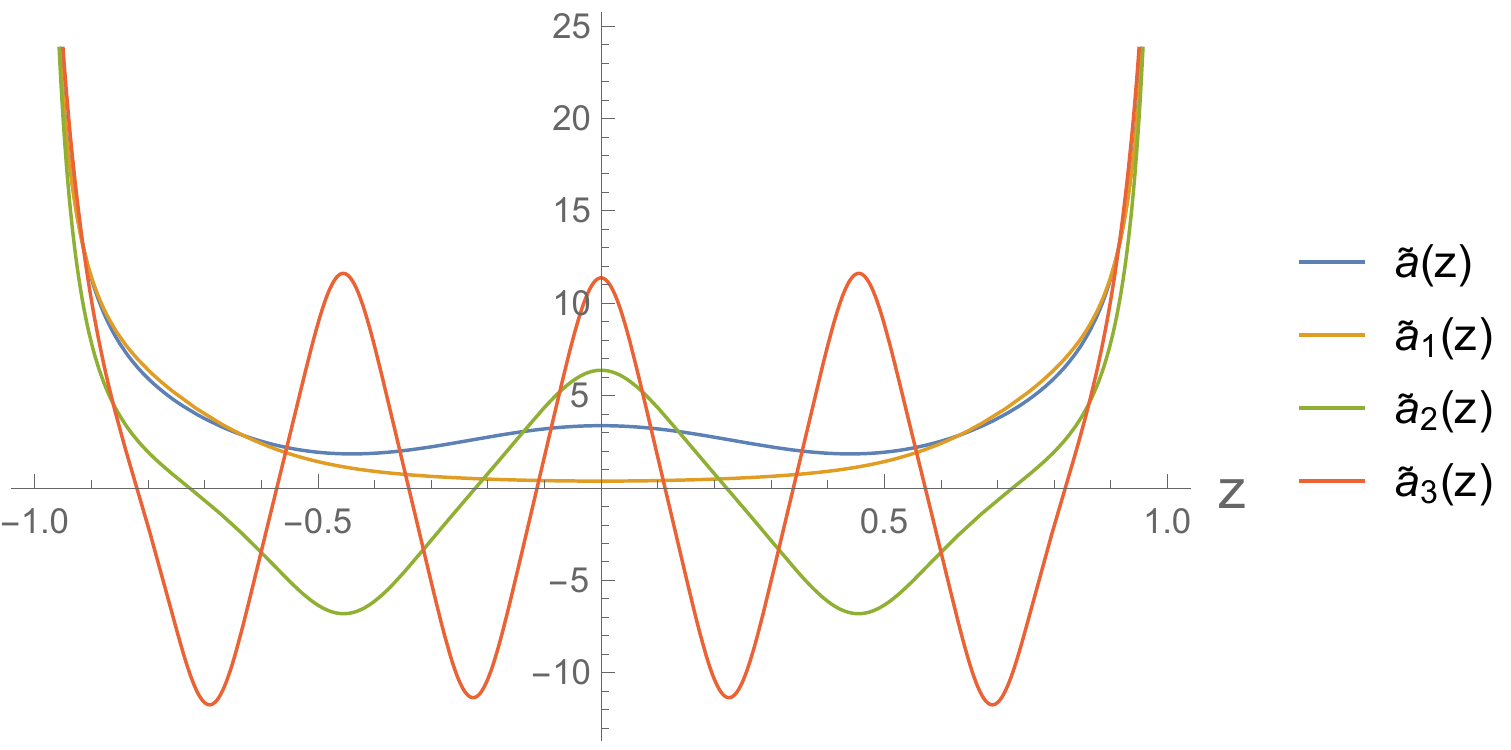}
    \caption{Multiple solutions of the ODE (\ref{ODE}) with potential $V(z)$ defined by the scale factor (\ref{testscale}), with $z_0=1$ and $mL=3/2$. The solutions displayed are the original scale factor (\ref{testscale}) $\tilde{a}(z)$, and additional solutions of the form $\tilde{a}_i(z)=\tilde{a}(z)+f_i(z)$, where $f_i(z)$ is a solution of (\ref{addsol}). In particular, they are solutions of (\ref{addsol}) with initial conditions $f_1(0)=-3$, $f_2(0)=3$, $f_3(0)=8$. The solutions $\tilde{a}_i(z)$ are symmetric and display the asymptotic behavior (\ref{limit}). In general, there are infinite solutions of this form.}
    \label{fig:multsol}
\end{figure}
Therefore, the solution of the ODE (\ref{ODE}) is not uniquely determined by symmetry and by the asymptotic behavior (\ref{limit}). We can conclude that, for $m^2>0$ and without further assumptions on the properties of the scale factor, we cannot reconstruct the wormhole geometry from a single potential $V(z)$, i.e. from a single scalar field spectrum\footnote{If we assume that the regular part of the scale factor $\tilde{a}(z)$ has a convergent asymptotic expansion around $z=\pm z_0$, and that such an expansion converges to $\tilde{a}(z)$ itself, then we can reconstruct the scale factor uniquely by matching its asymptotic expansion with an asymptotic expansion for the reconstructed potential. Numerically, only the first few terms are needed in order to obtain initial conditions at a point $\tilde{z}=-z_0+a$ (with small but finite $a$), and then the ODE (\ref{ODE}) can be solved with such initial conditions. Note that the assumption of a convergent asymptotic series alone is not enough: there might be non-perturbative contributions to the scale factor $\tilde{a}(z)$ (e.g. of the form included in equation (\ref{testscale})) which cause the asymptotic expansion to converge to a function $g(z)\neq \tilde{a}(z)$. When this is the case, the asymptotic expansion leads to a wrong reconstructed scale factor.}.  

Finally, for $\alpha<2$ (tachyonic bulk scalar field, $m^2<0$), our numerical analysis suggests that there is no solution of (\ref{addsol}). Although we have no mathematical proof for the non-existence of a solution of the problem (\ref{addsol}), the numerical evidence and the fact that multiple solutions of the ODE (\ref{ODE}) exist for any $\alpha>2$ but a unique solution exists for $\alpha=2$ points towards the existence of a unique symmetric solution of the ODE (\ref{ODE}) with the correct asymptotics (\ref{limit}) when $\alpha<2$. Therefore, we expect that the spectrum of a single tachyonic scalar field is sufficient to reconstruct the wormhole geometry.

From the discussion above we can conclude that, without making additional assumptions about the sign of $m^2$ or the analyticity of the scale factor, it is not possible to reconstruct the wormhole geometry from a single scalar spectrum. In particular, although for $m^2\leq 0$ one spectrum is sufficient to reconstruct the scale factor, for $m^2>0$ additional data is required to perform the task. 

To understand this better, it is useful to look at how multiple solutions differ from each other asymptotically, when they exist (i.e. for $m^2>0$). This means to study the form of equation (\ref{addsol}) when keeping only leading terms near one of the boundaries. Using the ansatz (\ref{ansatz}) for $\tilde{a}_1$ and focusing on the left boundary at $z=-z_0$, we get
\begin{equation}
f''(z)+\frac{2\alpha-6}{(z+z_0)^2}f(z)=0.
\label{asymptf}
\end{equation}
Using a power law ansatz $f(z)\sim (z+z_0)^\gamma$ in equation (\ref{asymptf}) we find
\begin{equation}
    \gamma_\pm=\frac{1}{2}\pm \frac{\sqrt{9-8(\alpha-2)}}{2}
    \label{exponent}
\end{equation}
where we remind that $\alpha-2=m^2L^2$ and that $f(z)$ must be regular at the boundary (i.e. only $\Re(\gamma)\geq 0$ leads to meaningful solutions). An analogous result can be obtained for the right asymptotic boundary at $z=z_0$.

For the $m^2>0$ case of interest (where multiple solutions to the ODE (\ref{ODE}) exist), we can distinguish three cases:
\begin{enumerate}
    \item For $2<\alpha< 3$ ($0<m^2L^2< 1$) or $\alpha=25/8$ ($m^2L^2=9/8$), there is only one non-negative real value of $\gamma$. Therefore in general we get the asymptotic behavior $f(z)\sim c_+(z+z_0)^{\gamma_+}$. This represents a one-parameter family of solutions, and a given solution is determined by a specific value of $c_+$. 
    \item For $3\leq \alpha < 25/8$ ($1\leq m^2L^2 <9/8$), there are two non-negative real values of $\gamma$. The asymptotic behavior of the general solution is then given by $f(z)\sim c_+(z+z_0)^{\gamma_+}+c_-(z+z_0)^{\gamma_-}$, and we have a two-parameter family of solutions parametrized by $(c_+,c_-)$.
    \item For $\alpha > 25/8$ ($m^2L^2>9/8$), there are two complex values of $\gamma$ with positive real part $\Re(\gamma)=1/2$. The corresponding asymptotic behavior of the general solution can be written as
    \begin{equation}
         f(z)\sim c_1 (z+z_0)^{\frac{1}{2}}\cos\left[\frac{\sqrt{9-8(\alpha-2)}}{2}\log(z)\right]+c_2 (z+z_0)^{\frac{1}{2}}\sin\left[\frac{\sqrt{9-8(\alpha-2)}}{2}\log(z)\right].
    \end{equation}
   This represents a two-parameter family of solutions, parametrized by $(c_1,c_2)$.
\end{enumerate}
Since the near-boundary limit of the bulk theory corresponds to the UV limit of the dual microscopic field theory, we expect the short-distance expansions of various field theory quantities (short-distance two-point functions or regularized entanglement entropies for small regions) should contain enough information about the near-boundary behavior of the scale factor to distinguish the possible solutions. In particular, it should be possible to obtain the value of the parameters $c_\pm$ (or $c_{1,2}$) from such short-distance behavior. This would allow one to uniquely reconstruct the wormhole geometry for $m^2>0$ using a single scalar spectrum and the UV behavior of the corresponding boundary scalar two-point function. We leave a detailed analysis of this possibility to future work.

Taking a different route, we will now show that if we have access to two scalar field spectra for two non-interacting bulk scalar fields of different mass (for any sign of $m^2$), the wormhole geometry can be uniquely reconstructed. Note that two spectra are always sufficient to reconstruct the scale factor, but the discussion in the previous paragraphs suggests that they might not be necessary. 

\subsubsection*{Reconstructing the wormhole from two scalar spectra}

Let us now assume that we have access to two distinct mass spectra of scalar particles in the dual confining gauge theory, corresponding to two non-interacting scalar fields with different masses in the same wormhole geometry. Suppose that, given the knowledge of the two spectra, we reconstructed the two associated potentials $V_1(z)$ and $V_2(z)$ using the procedure described in the previous subsections. We can then evaluate the ODE (\ref{ODE}) for the two potentials at $z=0$, and take their difference. Since $\tilde{a}(z)$ is the same in both equations, and we can obtain the values of $\alpha_1$, $\alpha_2$ from the asymptotic spectra, this yields an initial condition for the scale factor at the center of the wormhole:
\begin{equation}
    \tilde{a}(0)=\tilde{a}_0\equiv \sqrt{\left|\frac{V_1(0)-V_2(0)}{\alpha_1-\alpha_2}\right|}.
    \label{a0}
\end{equation}
We can now solve the ODE (\ref{ODE}) for any of the two potentials imposing initial conditions at the center of the wormhole: $\tilde{a}(0)=\tilde{a}_0$, $\tilde{a}'(0)=0$. This uniquely determines the wormhole scale factor.

It is worth noting that this procedure holds for any form of the scale factor (even in the non-physical case where it is not positive everywhere) and for both signs of $m^2$. Moreover, since in generic settings we expect to have multiple scalar fields in the wormhole geometry \cite{Antonini2022,Antonini2022short}, it is reasonable, and in fact to be expected, that we can have access to multiple spectra of scalar particles in the dual confining gauge theory.

\subsection{Wormhole reconstruction algorithm}\label{sec:reconstruction_algorithm_summary}

Let us summarize the algorithm that allows us to reconstruct the wormhole geometry from two scalar mass spectra in the underlying microscopic theory.\footnote{As discussed above, if we have access to a spectrum associated with a bulk scalar field with $m^2\leq 0$, there is no need for a second spectrum. We can therefore skip step 10, and solve the ODE (\ref{ODE}) imposing that the scale factor is symmetric and has the correct asymptotics (\ref{limit}).} These spectra may be given as input data, or they may be obtained from a knowledge of the microscopic correlators, as explained in Section \ref{sec:spectrum_from_correlator}. We will assume that the spectra are of the asymptotic form $\lam_j\sim Z j^2+Aj+B\log j+C$, with $Z,A>0$ and $B,C$ arbitrary. As we have seen, this is a necessary condition for a microscopic spectrum to have a dual description in terms of a free scalar field in a wormhole background. The reconstruction algorithm is as follows:

\begin{enumerate}
    \item Start with the first spectrum; call it $\lam_j$. Fit the asymptotic values of $\lam_j$ to the curve $Zj^2+Aj+B\log j+C$ and determine the constants $Z,A,B,C$.
    \item Using \eqref{eq:z_0_vs_Z}-\eqref{eq:c_vs_C} determine $z_0,\alpha,\beta,c$, and from \eqref{eq:Vt} construct $\Vt(z)$.
    \item Choose an integer $\jmax>0$. The wormhole scale factor will be accurately reconstructed on scales larger than $z_0/\jmax$. 
    \item Compute the lowest $\jmax+1$ eigenvalues $\lamt_0,\lamt_1,\dots,\lamt_{\jmax}$ by solving the Schr\"odinger equation
    \begin{align}
        -u''(z)+\Vt(z)u(z)=\lamt_j u(z),
    \end{align}
    subject to normalizable boundary conditions, $\um(\pm z_0) = 0$. We are using $\up$ and $\um$ to denote the normalizable and non-normalizable components to mode functions; see \eqref{eq:u_asymptotics}.
    \item Define $\Dbar \equiv \sqrt{\alpha+1/4}$. For $0\le j \le \jmax$ compute $\ut_j(z)$ by solving
    \begin{align}
        \ut_j''(z)=\big[\Vt(z)-\lam_j\big]\ut_j(z),
    \end{align}
    subject to the initial conditions
    \begin{align}
        \ut_j^{(+)}(-z_0)&=\frac{1}{2\Dbar},\\
        \ut_j^{(-)}(-z_0)&=0.
    \end{align}
    \item For $0\le j \le \jmax$, obtain $\omega'(\lam_j)$ from \eqref{eq:omega_prime_approx}.
    \item For $0\le j \le \jmax$, define $\yt_j(z)\equiv 2\frac{\ut_j(-z)-(-1)^j \ut_j(z)}{\omega'(\lam_j)}$.
    \item Determine $u_0,u_1,\dots,u_{\jmax}$ by solving the system of equations
    \begin{align}
        u_j''(z)&=\left(\Vt(z)+\sum_{i=0}^{\jmax}\big[\yt_i(z)u_i(z)\big]'-\lam_j\right)u_j(z)
    \end{align}
    subject to the boundary conditions
    \begin{align}
        u_j^+(-z_0)&=\frac{1}{2\Dbar},\\
        u_j^-(-z_0)&=0.
    \end{align}
    \item Obtain the potential by setting
    \begin{align}
        V(z)=\Vt(z)+\frac{1}{2}
        \sum_{j=0}^{\jmax}\Big(\big[\yt_j(z)u_j(z)\big]'+\big[\yt_j(-z)u_j(-z)\big]'\Big).
    \end{align}
    Here we have ensured by hand that the reconstructed $V(z)$ is symmetric, since this is not automatically true for finite $\jmax$.\footnote{We would like the approximate potential to be symmetric so that the resulting scale factor in the wormhole picture is symmetric and thus the analytically continued scale factor in the cosmology picture is real.}
    \item Repeat steps 1-9 to reconstruct the potential associated with the second spectrum. Now we have reconstructed two potentials, $V_1(z)$ and $V_2(z)$, from the respective asymptotic spectra. The value of the scale factor at the center of the wormhole is then given by
    \begin{equation}
        \tilde{a}(0)=\tilde{a}_0\equiv \sqrt{\left|\frac{V_1(0)-V_2(0)}{\alpha_1-\alpha_2}\right|}
    \end{equation}
    where $\tilde{a}(z)=a(z)/L$.
    \item Compute the wormhole scale factor by solving
    \begin{equation}
    \begin{cases}
    \tilde{a}''(z)+(\alpha-2)\tilde{a}^3(z)-V(z)\tilde{a}(z)=0\\
    \tilde{a}(0)=\tilde{a}_0\\
    \tilde{a}'(0)=0
    \end{cases}
    \end{equation}
    \item Finally, if one is interested in the corresponding FRW scale factor, it can be obtained by analytic continuation. Numerically this can be approximated by flipping the sign of every other non-zero term in the Taylor expansion of $\tilde{a}(z)$ around $z=0$. Namely,
    \begin{align}
        \afrw(t) \approx \sum_n \frac{(-1)^{n/2}}{n!}\frac{d^{2n}\tilde{a}(z)}{dz^{2n}}\Big|_{z=0}t^{2n}.
    \end{align}
\end{enumerate}

\subsection{Numerical example of wormhole reconstruction}
\label{sec:reconstruction_example}

In this subsection we provide an explicit example of implementation of the reconstruction algorithm summarized in Section \ref{sec:reconstruction_algorithm_summary}. The scale factor we wish to reconstruct has the analytic form\footnote{In this explicit example we did not include non-perturbative terms for the sake of numerical efficiency. Note that this scale factor does not have any specific physical meaning, and in particular it does not analytically continue to a Big Bang-Big Crunch cosmological scale factor. However, it is simple enough for the reconstruction to not be computationally too expensive, while its behavior is non-trivial enough to make the reconstruction example significant.}
\begin{equation}
    \tilde{a}(z)=\frac{1}{z_0+z}+\frac{1}{z_0-z}+2\cos\left(\frac{2\pi z}{z_0}\right)
    \label{eq:scalenum}
\end{equation}
where we set $z_0=1$. As a first step, we need to compute the asymptotic spectra of two potentials $V_1(z)$ and $V_2(z)$ of the form
\begin{equation}
    V(z)=\frac{a''(z)}{a(z)}+m^2L^2a^2(z)
    \label{eq:pot2}
\end{equation}
for two different choices of $mL$, and the first few eigenvalues of the same potentials, i.e. $\lambda_j^{(1)}$, $\lambda_j^{(2)}$ for $j<j_{max}$. This data will serve as input for the reconstruction procedure. Our choices of masses is $m_1L=1/2$, $m_2L=3/2$, and we choose $\jmax=50$, i.e. we reconstruct the potentials from their first 51 eigenvalues. The first step in our algorithm is to fit the asymptotic spectra as in equation (\ref{eq:fitformula}).
\begin{figure*}[h]
    \centering
    \subfloat[]{
        \includegraphics[width=0.45\linewidth]{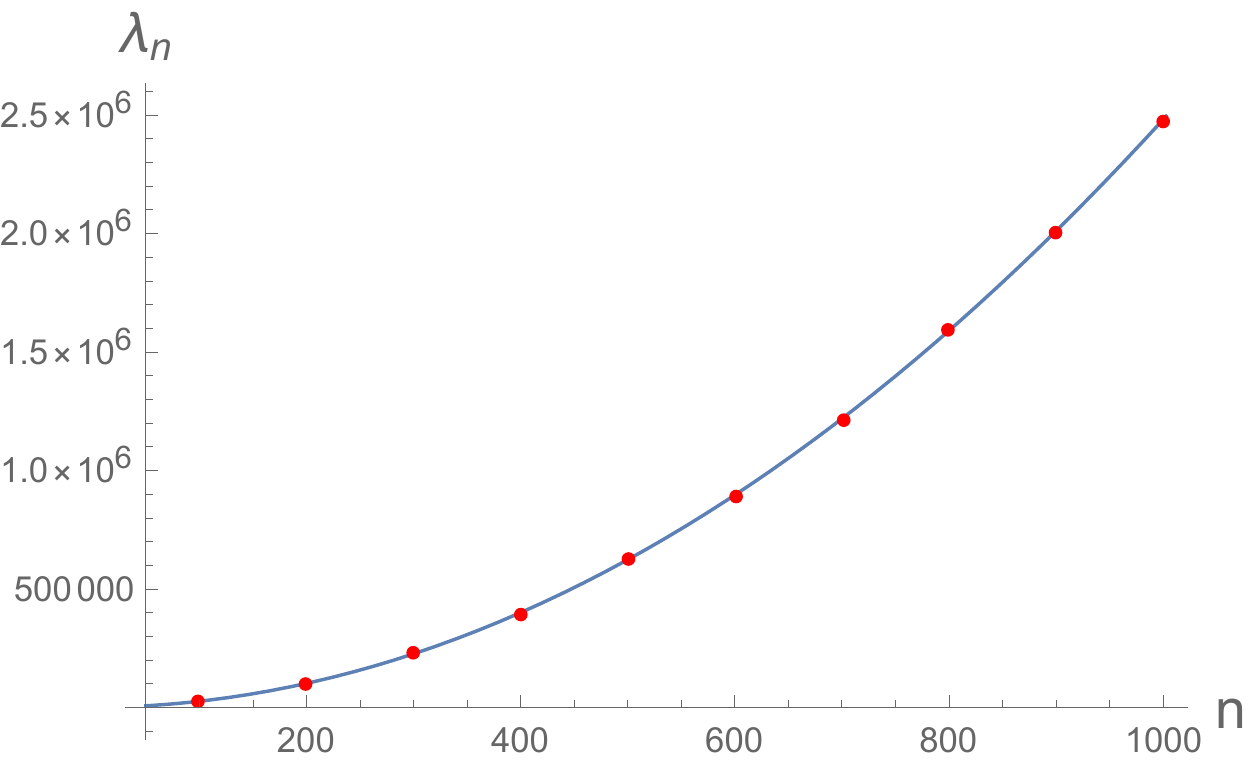}}
    \subfloat[]{
        \includegraphics[width=0.45\linewidth]{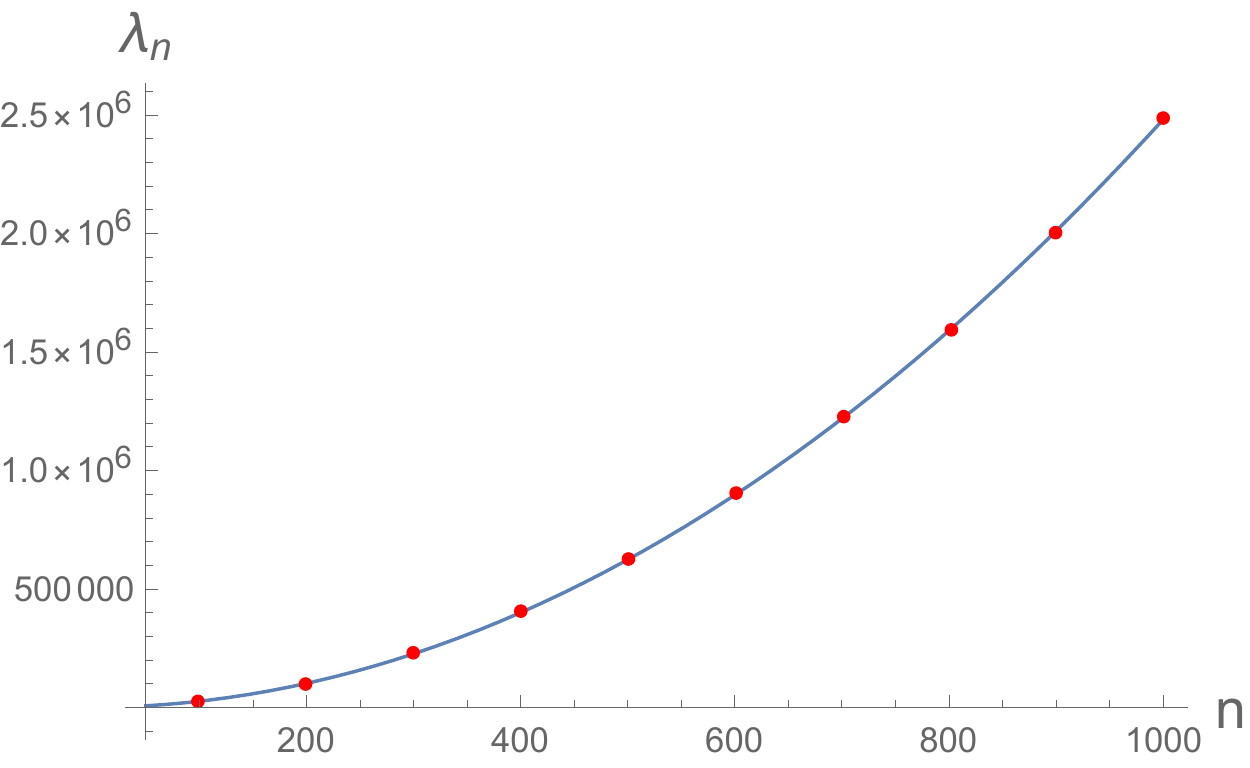}}
       \caption{Asymptotic spectra generated by the potentials (\ref{eq:pot2}) for $\tilde{a}(z)$ given by (\ref{eq:scalenum}). The fitted data also include all eigenvalues with $j\in [80,99]$ and $j\in [901,999]$, omitted in the plots. (a) $mL=1/2$. The result of the fit is $Z_1=2.467401$, $A_1=10.270160$, $B_1=-3.737523$, $C_1=39.926397$. (b) $mL=3/2$. The result of the fit is $Z_2=2.467401$, $A_2=12.934720$, $B_2=6.267183$, $C_2=36.382785$.}
    \label{recpotential}
\end{figure*}
Using equations (\ref{eq:z_0_vs_Z}), (\ref{eq:a_vs_A}), (\ref{eq:b_vs_B}) and (\ref{eq:c_vs_C}), this yields $z_0=1$ for both spectra and the parameters of the test potentials (\ref{eq:Vt}):
\begin{equation}
    \begin{aligned}
        &\alpha_1=2.250097, \hspace{1cm} \beta_1=-3.737523, \hspace{1cm} c_1=33.996971;\\
         &\alpha_2=4.249159, \hspace{1cm} \beta_2=6.267183, \hspace{1cm} c_2=21.274552.\\
    \end{aligned}
    \label{eq:fitted_par}
\end{equation}
We can then compute the first 51 eigenvalues of the two test potentials defined as in (\ref{eq:Vt}), and apply the rest of the algorithm described in Section \ref{sec:reconstruction_algorithm_summary} to reconstruct the two potentials $V_1(z)$ and $V_2(z)$ from their spectra.
Plots of the reconstructed potentials, together with the respective test potentials (\ref{eq:Vt}) and the exact potentials computed directly from the scale factor (\ref{eq:scalenum}) are reported in Figure \ref{fig:recpot}. The reconstructed potentials match the exact potentials up to corrections of order $1/j_{max}\sim 10^{-2}$, as expected.

\begin{figure*}[h]
    \centering
    \subfloat[]{
        \includegraphics[width=0.45\linewidth]{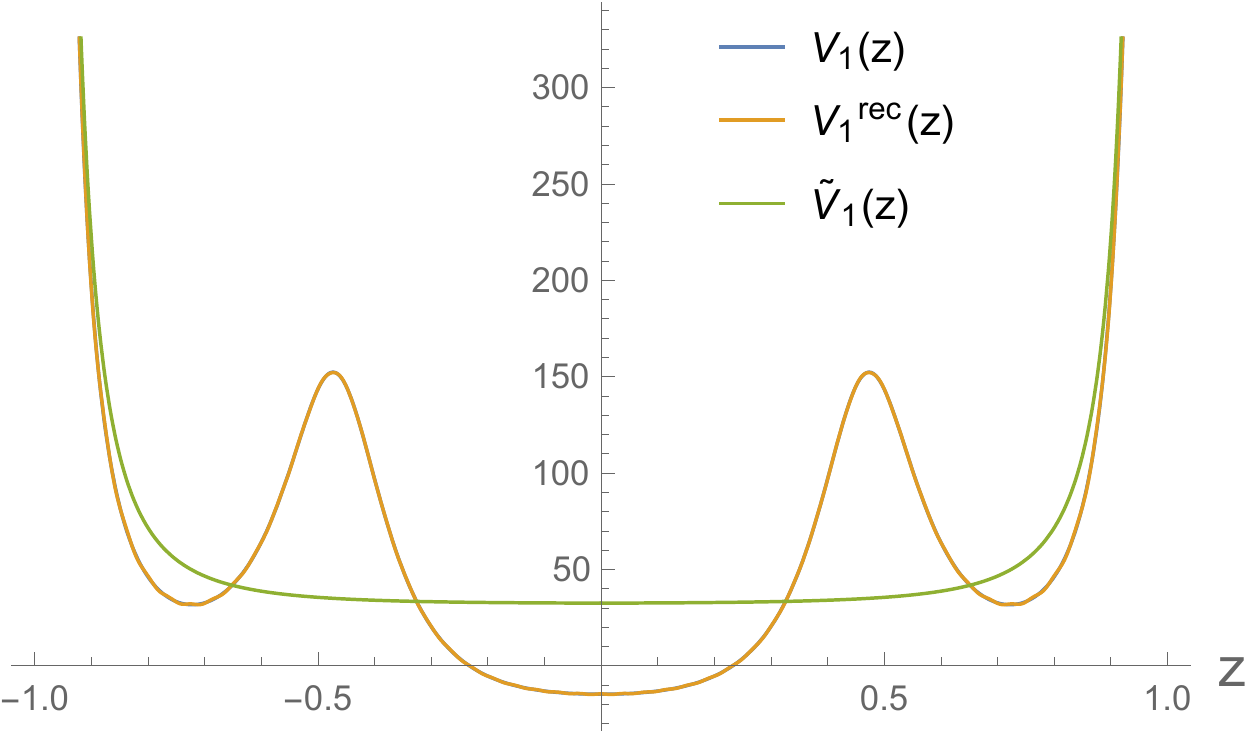}}
    \subfloat[]{
        \includegraphics[width=0.45\linewidth]{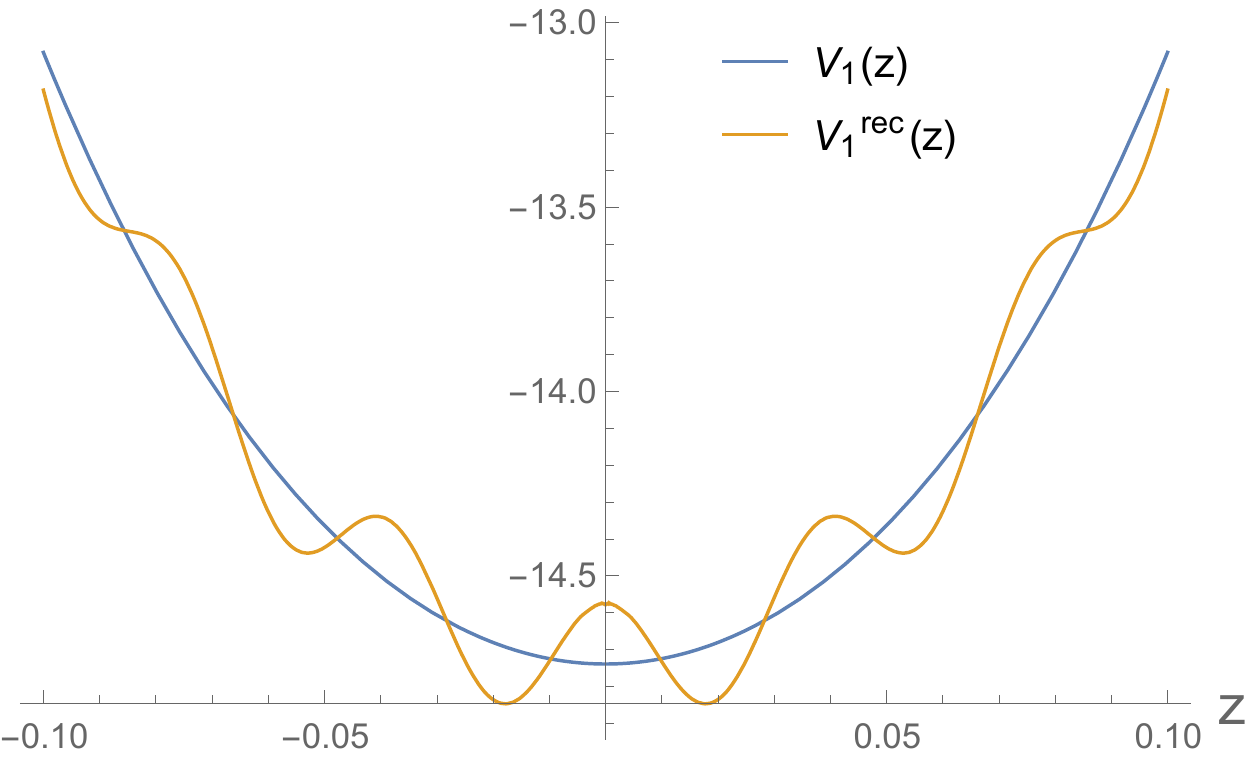}}\\
         \subfloat[]{
        \includegraphics[width=0.45\linewidth]{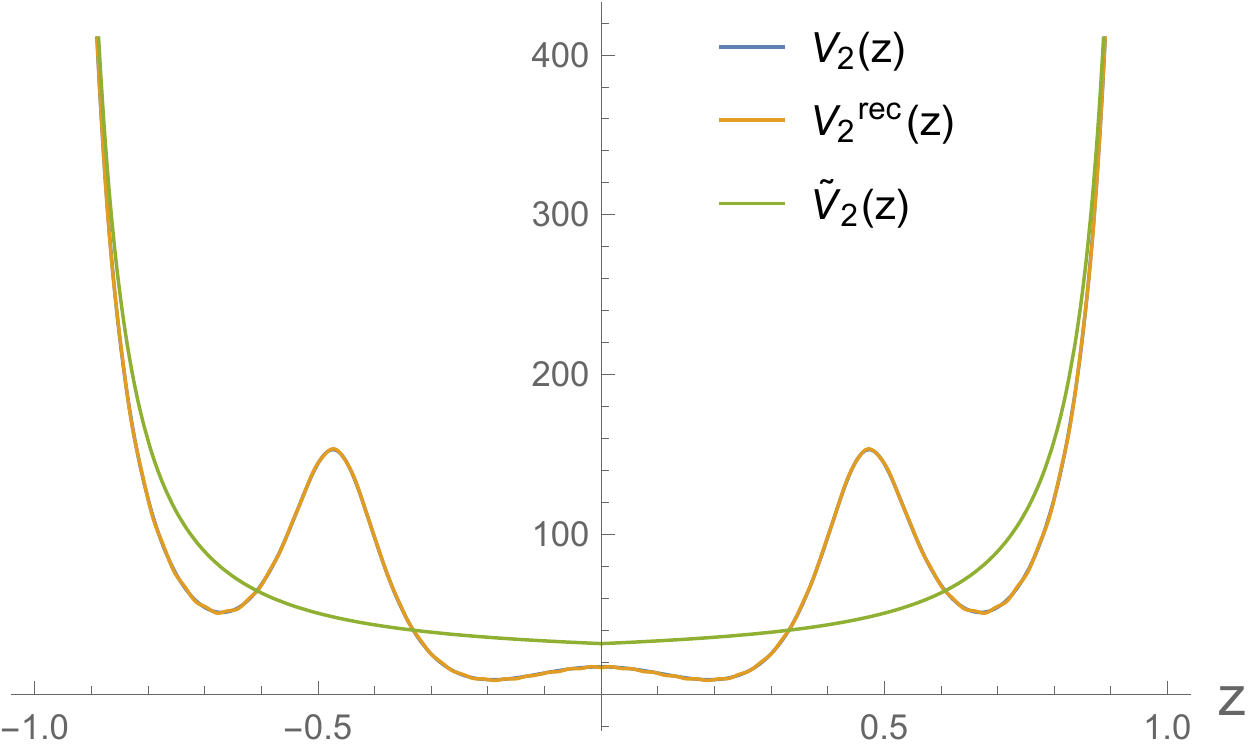}}
    \subfloat[]{
        \includegraphics[width=0.45\linewidth]{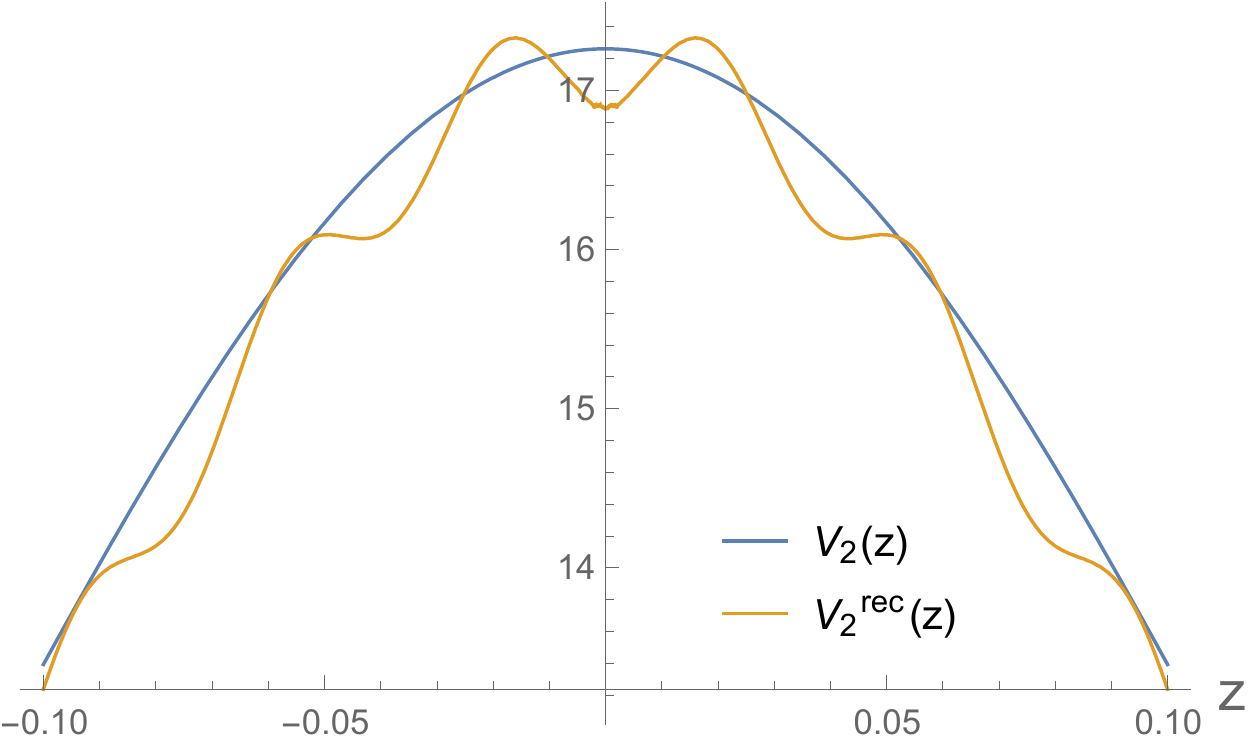}}
       \caption{True potentials $V_i(z)$ (computed by equation (\ref{eq:pot2}) with $\tilde{a}(z)$ given by (\ref{eq:scalenum})), reconstructed potentials $V^{rec}_i(z)$, and test potentials $\tilde{V}_i(z)$ (given by equation (\ref{eq:Vt}) with parameters (\ref{eq:fitted_par})) as a function of $z$. (a) $V_1(z)$, $V_2^{rec}(z)$, and $\tilde{V}_1(z)$ for $mL=1/2$. (b) Detail of $V_1(z)$ and $V_1^{rec}(z)$ around $z=0$ for $mL=1/2$. As expected, the reconstructed potential is a good approximation to the true potential up to corrections of order $1/j_{max}\sim 10^{-2}$. (c) $V_2(z)$, $V_2^{rec}(z)$, and $\tilde{V}_2(z)$ for $mL=3/2$. (d) Detail of $V_2(z)$ and $V_2^{rec}(z)$ around $z=0$ for $mL=3/2$. As expected, the reconstructed potential is a good approximation to the true potential up to corrections of order $1/j_{max}=\sim 10^{-2}$.}
    \label{fig:recpot}
\end{figure*}

We have now completed steps 1-9 of the reconstruction algorithm of Section \ref{sec:reconstruction_algorithm_summary} for two spectra, and reconstructed the two potentials $V_1(z)$ and $V_2(z)$ to good accuracy. The last two steps allow us to obtain the wormhole scale factor from such potentials. First, we must obtain the initial condition $\tilde{a}(0)=\tilde{a}_0$ as in equation (\ref{a0}). We could do this by directly subtracting $V_1^{rec}(0)$ from $V_2^{rec}(0)$, which would already allow us to obtain the initial condition to good accuracy, and obtain a good reconstruction of the wormhole scale factor. However, note that the reconstructed potentials have an oscillating behavior around $z=0$, because we only used a finite amount of eigenvalues to perform the reconstruction (see Figure \ref{fig:recpot}). Therefore, we can further improve the accuracy of the initial condition by averaging $V_1^{rec}(z)$ and $V_2^{rec}(z)$ over one oscillation period around $z=0$, and then using such average values in place of $V_1^{rec}(0)$, $V_2^{rec}(0)$ in (\ref{a0}). Using this procedure and the values (\ref{eq:fitted_par}) of $\alpha_1$ and $\alpha_2$, we obtain $\tilde{a}_0=3.989952$, where the exact value is $\tilde{a}_0^{true}=4$ (note that the error is again of order $1/j_{max}$, as expected). Finally, we can solve the ODE (\ref{ODE}) for either $V_1(z)$ or $V_2(z)$ with initial conditions $\tilde{a}(0)=\tilde{a}_0$, $\tilde{a}'(0)=0$, and obtain the wormhole scale factor. The numerically reconstructed $\tilde{a}_{rec}(z)$ is reported in Figure \ref{fig:scalerec} along with the exact scale factor (\ref{eq:scalenum}). The two agree with very good accuracy, up to corrections of order $1/j_{max}$.
\begin{figure*}[h]
    \centering
    \subfloat[]{
        \includegraphics[width=0.45\linewidth]{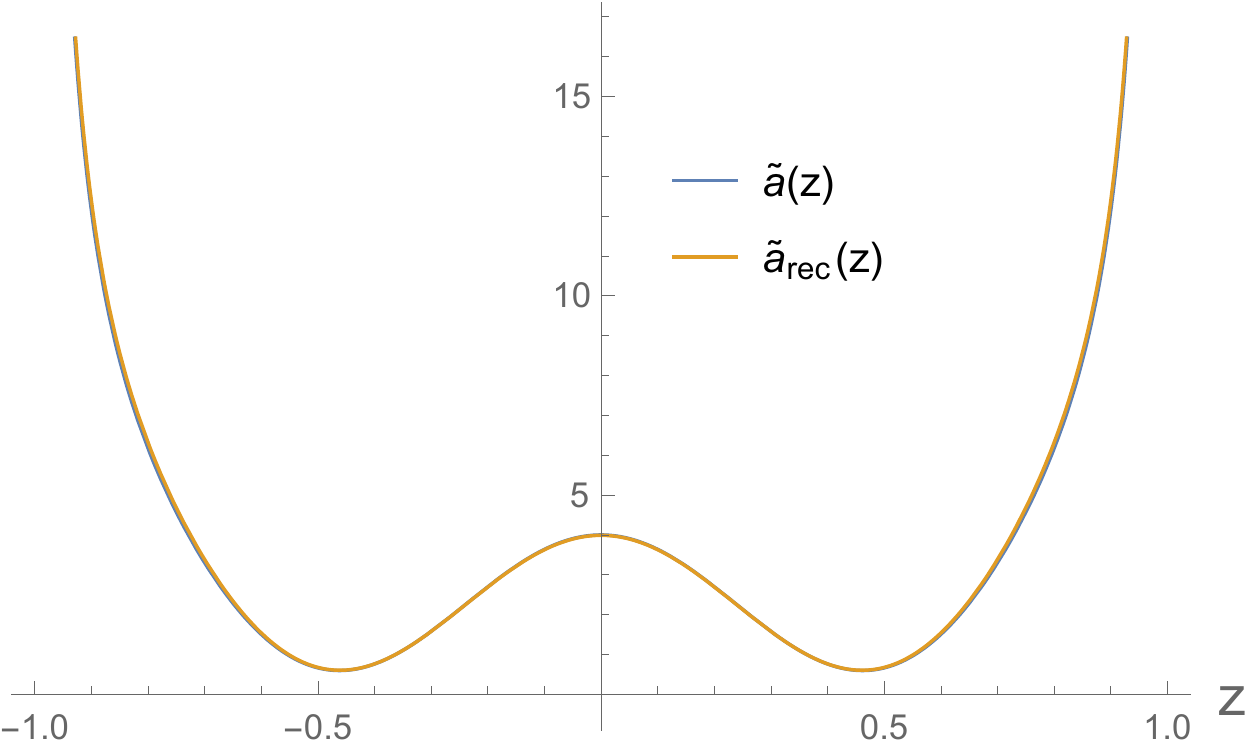}}
    \subfloat[]{
        \includegraphics[width=0.45\linewidth]{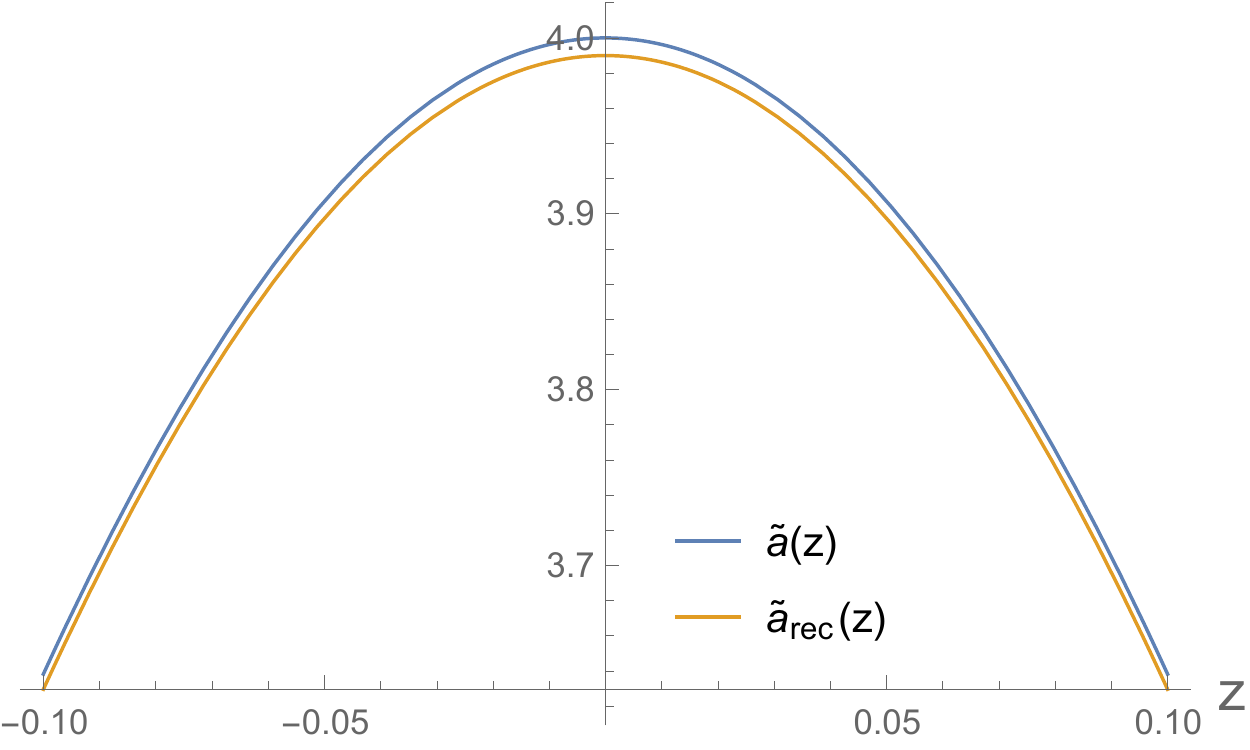}}\\
         \subfloat[]{
        \includegraphics[width=0.45\linewidth]{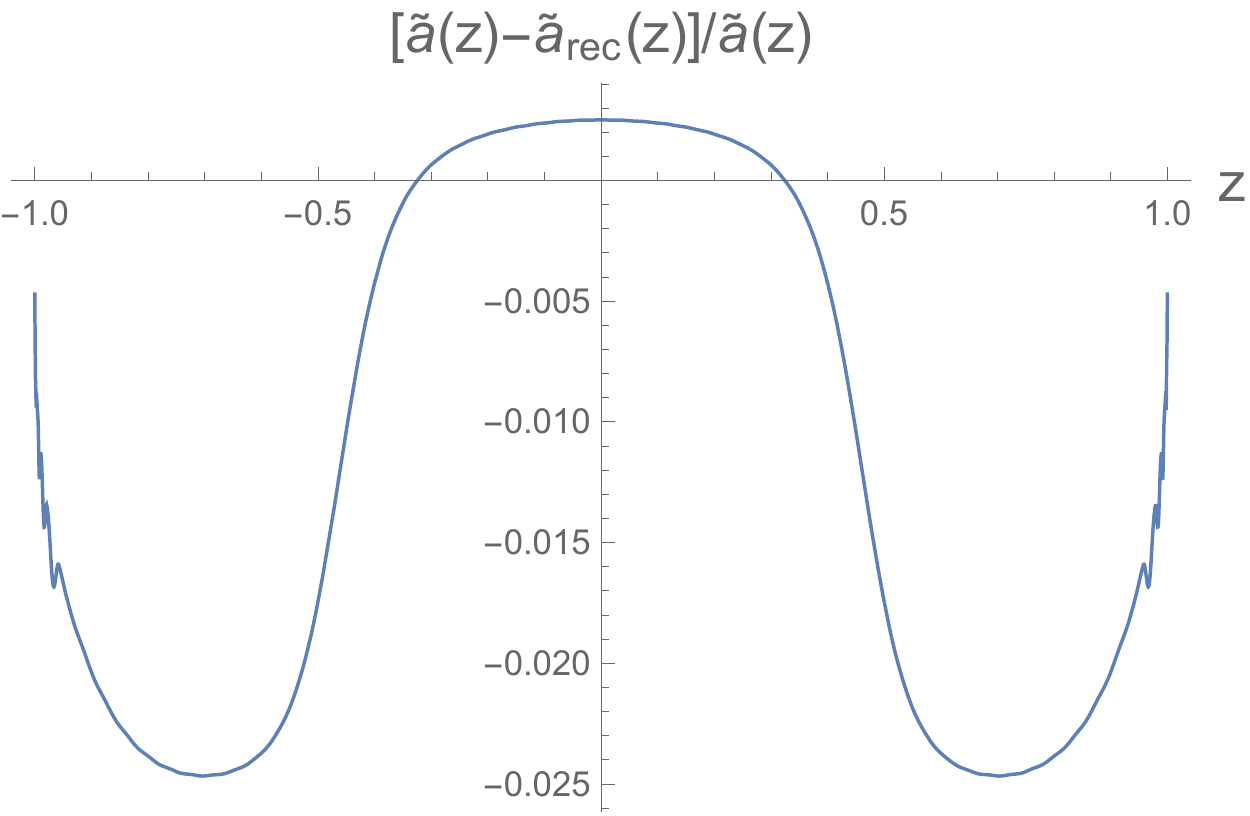}}
       \caption{True scale factor $\tilde{a}(z)$ (given by equation (\ref{eq:scalenum})) and reconstructed scale factor $\tilde{a}_{rec}(z)$ as a function of $z$. (a) The scale factors are almost indistinguishable at large scales. (b) Detail around the center of the wormhole at $z=0$. (c) Discrepancy between the true scale factor and the reconstructed scale factor, quantified by $[\tilde{a}(z)-\tilde{a}_{rec}(z)]/\tilde{a}(z)$. As expected, the reconstruction is accurate up to corrections of order $1/j_{max}\sim 10^{-2}$.}
    \label{fig:scalerec}
\end{figure*}

\subsection{Scale factor from heavy correlators}
\label{sec:metric_from_heavy}

In the previous sections we have shown how to reconstruct the wormhole geometry given certain information about the dual confining gauge theory. In particular we found that the reconstruction is possible if we are given two sets of microscopic correlators, associated with 3D CFT scalar operators $\mathcal O_1$ and $\mathcal O_2$ of different scaling dimensions. The duals of these operators are scalar fields of different masses living in the wormhole geometry. So far we have not assumed anything about the masses of these scalar fields, i.e. we have made no restrictions on the dimensions of the microscopic operators.

However, we will now see that the situation is much simpler if we have access to the two-point correlator of a heavy operator $\mathcal O$, i.e. an operator with a large scaling dimension. Namely we will find that a knowledge of this correlator is sufficient to reconstruct the wormhole geometry; there is no need for additional information from a different operator. 

The key observation which allows for such a simplification is to note that for a scalar field of large mass $m$ the spatial two-point function in the wormhole is given approximately by 
\be
\langle \phi(x) \phi(y) \rangle \sim e^{-m L(x,y)},
\ee
where $L$ is the geodesic distance between $x$ and $y$. Via the extrapolate dictionary, this leads to a spatial two-point function
\be\label{eq:heavy_correlator}
\langle {\cal O}(w) {\cal O}(0) \rangle \sim e^{-m L(w)}
\ee
for the heavy operator $\mathcal O$ dual to the field $\phi$, where $L$ is the regulated length of the geodesic. See \cite{Faulkner:2018faa} for a review. Since geodesics are manifestly associated with the geometry of the spacetime, it seems plausible that a knowledge of various geodesic lengths via the correlators of $\mathcal O$ is enough information to explicitly reconstruct the wormhole. We will now show that this is indeed the case. The discussion is similar to \cite{Bilson:2010ff} though our final approach is somewhat different and we get a more explicit result.

Let us begin by expressing the wormhole metric in the form
\be
\label{genmet}
ds^2 = A^2(z) dz^2 + B^2(z) dx_\mu dx^\mu .
\ee
This can be mapped to the familiar conformally flat gauge by changing to a coordinate $\tilde z$ defined by $Adz=Bd\tilde z$. Letting $x$ be one of the spatial coordinates of the field theory, we have that spatial geodesics $z(x)$ lying in the $z-x$ plane are determined by extremizing the action
\be
\label{tpact}
S = \int dz \sqrt{A^2(z) +  B^2(z) \left( \frac{dx}{dz} \right)^2  }  \; .
\ee
The conserved quantity associated with translation invariance in the $x$ direction is
\be
\label{eq:mom}
p_x = \frac{\partial\cal L}{\partial x'} = \frac{B^2(z)}{\sqrt{ A^2(z) \left( \frac{dz }{dx} \right)^2 + B^2(z) }}.
\ee
While the action itself (giving the length of a geodesic) diverges, the variation of the action with respect to one of the endpoints is finite. We recall that the variation of the action about an on shell configuration with respect to an endpoint variation $\delta x_i$ in the direction $x_i$ can be expressed as
\be
\delta S = p_i \delta x_i \;.
\ee
Taking $x_i$ to be the direction $x$ we find
\be
\frac{\delta S}{\delta x} = p_x.
\ee

For geodesics with two endpoints at the same asymptotic boundary separated by $w$, there will be some deepest point $z_*(w)$ to which the geodesic penetrates. At this point, $dz/dx = 0$, so from \eqref{eq:mom} we have that $p_x = B(z_*)$. Thus, we have the relation
\be
\label{vareq}
B(z_*) = \frac{\delta S}{\delta x} \equiv S_1(w).
\ee
The right side here can be computed from CFT data. Namely, the regulated version of the geodesic length $S$ is given by $L(w)$ and so using the heavy correlator \eqref{eq:heavy_correlator} we obtain
\be\label{eq:S1}
S_1(w) = -\frac{1}{m} \frac{d}{dw} \ln \langle {\cal O}(w) {\cal O}(0) \rangle .
\ee

Now, using $p_x=B(z_*)$ and \eqref{eq:mom} gives us an equation for $dx/dz$ which can be integrated along half the geodesic curve to obtain another relation between $w$ and $z_*$. Namely
\be
\label{inteq}
w/2 = \int_{z_\infty}^{z_*} dz \frac{A(z)}{B^2(z)} \frac{1}{\sqrt{\frac{1}{B^2(z_*)} - \frac{1}{B^2(z)}}},
\ee
where $z_\infty$ is the asymptotic value of the coordinate $z$.\footnote{This assumes that $z$ increases toward the interior. Otherwise the limits of integration would be reversed.} By a choice of gauge, we can take $B(z)$ to be any convenient function. In this case, the CFT data $S_1(w)$ and the known function $B(z)$ can be used in equation (\ref{vareq}) to write $w$ on the left side of (\ref{inteq}) as some function of $z_*$. The equation (\ref{inteq}) can then be understood as an integral equation (specifically, a weakly singular Volterra equation of the first kind) for the remaining undetermined metric function $A(z)$.

The integral equation
\be
f(s) = \int_a^s \frac{y(t) dt}{\sqrt{g(s) - g(t)}} 
\ee
for $y(t)$ has solution
\be
y(s) = \frac{1}{\pi} \frac{d}{ds} \left[ \int_a^s dt  \frac{f(t) g'(t)}{\sqrt{g(s) - g(t)}} \right].
\ee
Applying this to our case, we have that
\be\label{eq:A(z)}
A(z) = -\frac{1}{\pi} B^2(z) \frac{d}{dz} \left[ \int_{z_\infty}^{z} d \hat{z}  \frac{w(\hat{z}) B(z) B'(\hat{z})}{B^2(\hat{z}) \sqrt{B^2(\hat{z}) - B^2(z)}} \right]
\ee
Via this equation and (\ref{vareq}), we are able to fully determine the metric in terms of the heavy microscopic two-point function.

We can make things more explicit by choosing the function $B(z)$ to be equal to $S_1(z)$. In this case, from \eqref{vareq} we have
\be
S_1(z_*) =  S_1(w),
\ee
so we have $z_*=w$. Thus, our choice of $B$ corresponds to taking the $z$ coordinate of a point in the wormhole to be the width of a CFT interval whose geodesic penetrates to that point. Then, from \eqref{eq:A(z)} we have
\be
A(w) = -\frac{1}{\pi} S_1^2(w) \frac{d}{dw} \left[ \int_{0}^{w} d \hat{w}  \frac{\hat{w} S_1(w) S_1'(\hat{w})}{S_1^2(\hat{w}) \sqrt{S_1^2(\hat{w}) - S_1^2(w)}} \right],
\ee
and the metric is given by
\be\label{eq:heavy_metric}
ds^2 = A^2(w) dw^2 + S_1^2(w) dx_\mu dx^\mu.
\ee
In this way, we have explicitly expressed the metric in terms of the microscopic data encoded in the function $S_1(w)$ via equation \eqref{eq:S1}.

The form \eqref{eq:heavy_metric} of the reconstructed metric suggests the change of variables 
\be
\label{cov}
\mu(w) = -\frac{1}{\pi} \left[ \int_{0}^{w} d \hat{w}  \frac{\hat{w} S_1(w) S_1'(\hat{w})}{S_1^2(\hat{w}) \sqrt{S_1^2(\hat{w}) - S_1^2(w)}} \right] \; .
\ee
In terms of this radial coordinate, we have simply
\be
ds^2 = S_1^4(w(\mu)) d\mu^2 + S_1^2(w(\mu)) dx_\mu dx^\mu.
\ee

\subsection{Metric from entanglement entropy}
\label{sec:metric_from_entanglement}

Via a very similar analysis, we could also extract the metric from the entanglement entropy for a strip-shaped subsystem of one of the 3D CFTs.\footnote{To be precise, we should also include some part of the 4D CFT here, but since this has many fewer local degrees of freedom, we expect that the precise choice if this region shouldn't matter too much.}

We consider the entanglement entropy of a strip of width $w$ in one of the 3D CFTs. We assume that the result is well approximated by the area of a RT surface in the dual 4D traversable wormhole.\footnote{This expectation could be modified significantly if the volume of the internal space present in the UV-complete realization of our setup changes with radial position.}

The entanglement entropy has the usual UV divergence but also an IR divergence due to the infinite volume of the strip. We consider the quantity
\be
s_1(w) = \frac{d}{d w} s(w)
\ee
where $s(w)$ is the entanglement entropy per unit of length of the strip. In this case, it will be convenient to take the metric as
\be
\label{genmet2}
ds^2 = \frac{A^2(z)}{B(z)} dz^2 + B(z) dx_\mu dx^\mu \; ,
\ee
where $A$ and $B$ are in general different than the similarly named functions in Section \ref{sec:metric_from_heavy}. With this choice, the expression for the action  determining the extremal surface trajectory is precisely the same as the expression (\ref{tpact}) that we used in the previous section. Precisely the same analysis then tells us that
if we choose $B(z) = s_1(z)$, such that $z$ is identified with $w$ (i.e. the width of a strip whose RT surface barely reaches our point $z^*$), then the remaining metric function is 
\be
A(w) = -\frac{1}{\pi} s_1^2(w) \frac{d}{dw} \left[ \int_{0}^{w} d \hat{w}  \frac{\hat{w} s_1(w) s_1'(\hat{w})}{s_1^2(\hat{w}) \sqrt{s_1^2(\hat{w}) - s_1^2(w)}} \right].
\ee
However, differently from Section \ref{sec:metric_from_heavy}, the final metric is 
\be
\label{genmet4}
ds^2 = \frac{A^2(w)}{s_1(w)} dw^2 + s_1(w) dx_\mu dx^\mu .
\ee
In this case, a change of variables analogous to (\ref{cov}) with $S_1 \to s_1$ gives
\be
\label{genmet5}
ds^2 = s_1^3(w(\mu)) d\mu^2 + s_1(w(\mu)) dx_\mu dx^\mu .
\ee

\section{Discussion}
\label{sec:discussion}

In this paper we have investigated the relationship between observables in a bulk theory with an AdS planar eternal traversable wormhole background geometry and observables in the corresponding dual confining microscopic theory. We have studied what properties of the microscopic theory can be deduced from the existence of a bulk dual wormhole, and, conversely, how the wormhole geometry can be reconstructed from observables in the microscopic theory. 

The presence of the wormhole determines specific properties of the dual microscopic theory: spectrum of massive particles, existence of a massless sector, properties of two-point functions, and entanglement structure. In particular, the behavior of bulk quantum scalar and gauge fields implies the existence of a discrete spectrum of massive scalar and vector particles, along with a massless sector (associated with bulk gauge fields). 

On the other hand, certain observables in the microscopic confining theory (two-point functions of two 3D CFT scalar operators with different scaling dimensions, correlators of 3D CFT operators with large scaling dimensions, and entanglement entropies of subregions of the microscopic theory) allow one to reconstruct the dual wormhole geometry. A central result of the paper is the explicit algorithm we derived to reconstruct the wormhole metric from two-point functions of 3D CFT scalar operators.

Finally, although the results just outlined are interesting in their own right, we would like to also emphasize two consequences they have on the Big Bang-Big Crunch FRW cosmologies related to our wormhole theory by double analytic continuation \cite{VanRaamsdonk:2021qgv,Antonini2022,Antonini2022short}. 

First, the existence of a massless sector in the dual confining theory implies the existence of long-range correlations in its ground state, which was not obvious a priori. In the bulk effective field theory, this translates to correlations at every scale in the wormhole background along the non-compact directions. Since the $z=t=0$ codimension-2 surface is left invariant by the double analytic continuation relating the wormhole to the FRW cosmological universe \cite{Antonini2022,Antonini2022short}, this fact implies the existence of correlations at every scale in the cosmology at the (late) time-symmetric point where the universe stops expanding and starts re-collapsing. Therefore, the special state of the cosmology at the time-symmetric point (defined by the Euclidean path integral described in \cite{Antonini2022}) has built-in correlations even between regions that are never in causal contact at any point in the cosmological evolution. This feature helps to solve the cosmological horizon problem in terms of the properties of the underlying microscopic theory, without the need for inflation \cite{Antonini2022,Antonini2022short}.

Second, as we have already pointed out, the slicing duality implies that the FRW metric can be obtained from microscopic confining gauge theory observables by reconstructing the wormhole geometry, and then analytically continuing the resulting scale factor. Note that, if we do not make use of the slicing duality, the FRW universe is encoded in a very complex way into the physics of a specific highly excited state of only the microscopic 4D auxiliary degrees of freedom coupling the two 3D CFTs \cite{VanRaamsdonk:2021qgv,Antonini2022,Cooper2018,Antonini2019}. In fact, in a doubly holographic setup, where the auxiliary degrees of freedom are also holographic, the cosmology can be understood as living on an end-of-the-world (ETW) brane behind the horizon of a 5D black hole \cite{Cooper2018,Antonini2019}. Without relying on double holography, the cosmology can be seen as an ``entanglement island'' associated with large subregions of the 4D auxiliary system \cite{Antonini2022}. In both interpretations, reconstructing the cosmological evolution using observables in the excited state of the 4D auxiliary system requires computing expectation values of extremely complicated field theory operators\footnote{It has been shown in \cite{Cooper2018} that the entanglement entropy at early times of sufficiently large subregions of the boundary field theory can probe the cosmological scale factor. In fact, the RT surface associated with such subregions penetrates the black hole horizon and ends on the ETW brane, in analogy with the analysis of Hartman and Maldacena \cite{Hartman2013a}.}. Nonetheless, by means of the slicing duality, the cosmological evolution can be probed by reconstructing the corresponding wormhole geometry using confining field theory observables as simple as the spectrum of massive particles or two-point functions of operators with large scaling dimension, and then analytically continuing the scale factor. This fact shows explicitly how the slicing duality introduced in \cite{VanRaamsdonk:2021qgv,Antonini2022} allows one to relatively easily reconstruct holographic cosmologies which would naively be extremely complex to probe.
 
 \section*{Acknowledgements}

We would like to thank Rodrigo A. Silva and Wucheng Zhang for helpful discussions and comments. This work is supported in part by the National Science and Engineering Research Council of Canada (NSERC) and in part by the Simons foundation via a Simons Investigator Award and the ``It From Qubit'' collaboration grant. PS is supported by an NSERC C-GSD award. SA is partially supported by a Leon A. Herreid Science Graduate Fellowship Award. This work is partially supported by the U.S. Department of Energy, Office of Science, Office of Advanced Scientific Computing Research, Accelerated Research for Quantum Computing program ``FAR-QC'' (SA) and by the AFOSR under grant number FA9550-19-1-0360 (BS).

\appendix

\section{Quantization of a gauge field in the wormhole background}
\label{app:gauge_field}

Consider the action\footnote{In general, a boundary term dependent on the choice of boundary conditions is present in the action. We omit here, see \cite{Marolf:2006nd}} for a gauge field in our wormhole background\footnote{In this appendix we will use latin indices $I,J=0,1,2,3$ to indicate 4D components, and greek indices $\mu,\nu=0,1,2$ for 3D components $t,x,y$.}:
\begin{equation}
    S_{gauge}=-\frac{1}{4}\int d^4x\sqrt{-g}F_{IJ}F^{IJ}
\end{equation}
leading to the equations of motion (which, together with the Bianchi identity for $F_{IJ}$, provide the four Maxwell's equations)
\begin{equation}
    \partial_I\left(\sqrt{-g}F^{IJ}\right)=0.
\end{equation}
In our four-dimensional wormhole, these can be recast in the form
\begin{equation}
    \eta^{IK}\partial_I\partial_K A_J-\partial_J\left(\eta^{IK}\partial_I A_K\right)=0
\end{equation}
where $\eta_{IJ}$ is the 4D Minkowski metric. Note that, as a result of the conformal invariance of the gauge field action in 4D, the equations of motion have no dependence on the wormhole scale factor $a(z)$. With a slight abuse of notation, in the rest of this appendix we will raise and lower indices using the Minkowski metric, rather than the wormhole metric $g_{IJ}=1/a^2(z)\eta_{IJ}$. We will work in Lorenz gauge
\begin{equation}
   \partial^IA_I=0
   \label{lorenzgauge}
\end{equation}
in which the equations of motion decouple, reducing to four independent Klein-Gordon equations for a massless field in flat spacetime.

\subsection{Boundary conditions}
\label{app:gauge_bc}

For simplicity of notation, we define here $\ell=2z_0$ and shift $z\to z-z_0$ such that the two boundaries are located at $z=0,\ell$. The behavior of the gauge field near the two boundaries is given by\footnote{The argument $(x)$ stands for a dependence on the 3D coordinates $t,x,y$ only.}
\begin{equation}
\begin{aligned}
    &A_I(x,z\sim 0)\sim\alpha_I^{(L)}(x)(1+...)+z\beta_I^{(L)}(x)(1+...)\\
    &A_I(x,z\sim \ell)\sim\alpha_I^{(R)}(x)(1+...)+(\ell-z)\beta_I^{(R)}(x)(1+...)
    \end{aligned}
\end{equation}
where the dots indicate terms of higher order in $z$ and $\ell-z$ for the left and right boundaries respectively.
We must now specify a set of boundary conditions. We will do so for the $\mu$ components of the gauge field; we will see that the Lorenz gauge condition then uniquely determines the boundary condition for $A_z$. There are two standard choices for the boundary conditions\cite{Aharony:2010ay,Hijano:2020szl,Witten:2003ya,Yee:2004ju,Marolf:2006nd}: at each boundary we can either fix $\alpha_\mu(x)$ (Dirichlet boundary conditions), or fix $\beta_\mu(x)$ (Neumann boundary conditions), up to a (residual) gauge transformation\footnote{The Dirichlet and Neumann boundary conditions are sometimes referred to as ``magnetic'' and ``electric'' boundary conditions, because they correspond to fixing $F_{\mu\nu}$ and $F_{\mu z}$ at the boundary, respectively (although this nomenclature is somewhat improper, $z$ being a spatial direction).}. For later convenience, let us introduce a deformation term in the dual CFT action, given by 
\begin{equation}
    I_{d}=i\int d^3x O^{(\alpha)}_\mu(x)J^{\mu}(x),
    \label{deformation}
\end{equation}
where the role of $O^{(\alpha)}_\mu(x)$ and $J_\mu(x)$ depends on the choice of boundary conditions and will be clarified below.

Imposing Dirichlet boundary conditions, which can be reformulated in a gauge-invariant way by fixing $F_{\mu\nu}|_{\partial}$, the boundary value $\alpha_\mu(x)$ becomes a non-dynamical source $O^{(\alpha)}_\mu(x)$ in the dual 3D CFT living on the corresponding boundary, which couples to a global conserved current $J^\mu(x)$ with dimension $\Delta=2$, as in equation (\ref{deformation}). The one-point function of $J_\mu(x)$ is related to the coefficient $\beta_\mu(x)$ of the subleading term of the gauge field at the boundary. This leads to the ``standard quantization'', where the CFT path integral computes the generating functional $Z[O^{(\alpha)}]=\braket{\exp\left(i\int d^3x O^{(\alpha)}_\mu(x) J^\mu(x)\right)}$ with fixed sources $O^{(\alpha)}_\mu(x)=\alpha_\mu(x)$. Note that, since the current $J^\mu$ is conserved, the generating functional is insensitive to a gauge transformation $\alpha_\mu\to \alpha_\mu+\partial_\mu\lambda$, which shows explicitly why we can fix $\alpha_\mu$ up to a gauge transformation. In the corresponding bulk theory, the part of the gauge field giving rise to the subleading boundary term proportional to $\beta_\mu(x)$ is quantized. The boundary limit of its bulk correlators are dual to boundary correlators of the current $J_\mu(x)$.

The second choice, whose gauge-invariant form is given by fixing $F_{\mu z}|_{\partial}$\footnote{As we will see below, the gauge-invariant form of the boundary conditions is not completely equivalent to just fixing $\beta_\mu(x)$. In fact, differently from the Dirichlet case, fixing $\beta_\mu(x)$ and then imposing the Lorenz gauge condition (\ref{lorenzgauge}) does not completely determine a boundary condition for $A_z$. On the other hand, the (physically meaningful) gauge-invariant condition fixes the boundary conditions for $A_z$ completely. As far as the $\mu$ component of the gauge field is concerned, the two boundary conditions are equivalent.}, leaves the value $\alpha_\mu(x)$ of the gauge field at the boundary free to fluctuate. This must then be identified with the expectation value of a boundary vector operator $\braket{O^{(\alpha)}_\mu(x)}$ of dimension $\Delta=1$. In the CFT path integral, we must integrate over the field $O^{(\alpha)}_\mu(x)$. $J_\mu(x)$ appearing in (\ref{deformation}) must now be regarded as an external source, determined by the value of $\beta_\mu(x)$, which is now fixed.\footnote{Note that in the classical case the field $O^{(\alpha)}_\mu(x)$ has no kinetic term in the dual CFT's lagrangian. Therefore, from equation (\ref{deformation}), the equations of motion for $O^{(\alpha)}_\mu(x)$ imply $\braket{J^\mu}=0$, constraining the value of $\beta_\mu(x)$ to zero \cite{Aharony:2010ay}. In general, quantum corrections can introduce a kinetic term \cite{Aharony:2010ay}, allowing arbitrary values of the current, and therefore of $\beta_\mu(x)$.} 
This leads to the ``alternative quantization'', where the CFT path integral computes a generating functional $Z[J]=\braket{\exp\left(i\int d^3x O^{(\alpha)}_\mu(x) J^\mu(x)\right)}$ with fixed sources $J_\mu(x)$ determined by $\beta_\mu(x)$. Once again, the generating functional is invariant under gauge transformations. In the bulk theory, the part of the gauge field giving rise to the leading boundary term $O^{(\alpha)}_\mu(x)$ is quantized. For additional discussion on these points, see \cite{Aharony:2010ay,Hijano:2020szl,Witten:2003ya,Yee:2004ju,Marolf:2006nd}.

So far we did not make any assumption about the relationship between boundary conditions at the left ($z=0$) and right ($z=\ell$) boundaries. Since the microscopic construction dual to our traversable wormhole is reflection symmetric and involves two copies of the same 3D CFT \cite{Antonini2022}, we will assume that boundary conditions at the two boundaries are the same: either we fix $\alpha_\mu^{(L)}(x),\alpha_\mu^{(R)}(x)$ (i.e. $F_{\mu\nu}|_{z=0,\ell}$), or we fix $\beta_\mu^{(L)}(x),\beta_\mu^{(R)}(x)$ (i.e. $F_{\mu z}|_{z=0,\ell}$). Note that, in more general settings, different boundary conditions can be chosen at the two boundaries. The boundary conditions for $A_z$, which we will derive using the Lorenz gauge condition, also need to be the same at the two boundaries. Therefore, we can write in general 
\begin{equation}
    A_I(x,z)=A_I^{(1)}(x,z)+A_I^{(2)}(x,z)
    \label{fullfield}
\end{equation}
where 
\begin{equation}
    \begin{aligned}
    &A_I^{(1)}(x,z\sim 0)\sim\alpha_I^{(L)}(x), \hspace{1cm} A_I^{(1)}(x,z\sim \ell)\sim\alpha_I^{(R)}(x)\\
    &A_I^{(2)}(x,z\sim 0)\sim z\beta_I^{(L)}(x), \hspace{1cm} A_I^{(2)}(x,z\sim \ell)\sim (\ell-z)\beta_I^{(R)}(x)
    \end{aligned}
\end{equation}
where we omitted subleading corrections in $z$ and $\ell-z$. The standard quantization corresponds then to fixing $A_\mu^{(1)}(x,z)$ at $z=0,\ell$, and quantizing $A_\mu^{(2)}(x,z)$. The alternative quantization corresponds to fixing the value of $A_\mu^{(2)}(x,z)/z$ at $z=0$ and $A_\mu^{(2)}(x,z)/(\ell-z)$ at $z=\ell$, and quantizing $A_\mu^{(1)}(x,z)$. We can then write
$A_I^{(1)}(x,z)$ and $A_I^{(2)}(x,z)$ in the general form
\begin{equation}
    \begin{aligned}
   &A_I^{(1)}(x,z)=\sum_{n=0}^\infty\int \frac{dk_xdk_y}{2\pi\sqrt{\ell}\sqrt{\omega_n}\sqrt{1+\delta_{n,0}}}\left[\varepsilon_I^{(1)}(k_x,k_y,\rho_n)\cos(\rho_n z)\textrm{e}^{-i(\omega_n t-k_x x-k_y y)}+c.c.\right]\\
   &A_I^{(2)}(x,z)=\sum_{n=1}^\infty\int \frac{dk_xdk_y}{2\pi\sqrt{\ell}\sqrt{\omega_n}}\left[\varepsilon_I^{(2)}(k_x,k_y,\rho_n)\sin(\rho_n z)\textrm{e}^{-i(\omega_n t-k_x x-k_y y)}+c.c.\right]
    \end{aligned}
    \label{expansions}
\end{equation}
where the factor $\sqrt{1+\delta_{n,0}}$ is necessary to guarantee the normalization of the $n=0$ mode, $\omega_n=\sqrt{k_x^2+k_y^2+\rho_n^2}$, $\rho_n=n\pi/\ell$ is the momentum in the $z$ direction (whose discrete values are determined by the requirement of having the same kind of boundary conditions (Dirichlet or Neumann) at the two boundaries), and $\epsilon_I$ is the polarization 4-vector.

The gauge field (\ref{fullfield}) must satisfy the Lorenz gauge condition (\ref{lorenzgauge}). Using the expansion (\ref{expansions}), the Lorenz gauge condition takes the form
\begin{equation}
    \begin{cases}
    i\varepsilon_\mu^{(1)}k^\mu+\varepsilon_z^{(2)}\rho_n=0\\
    i\varepsilon_\mu^{(2)}k^\mu-\varepsilon_z^{(1)}\rho_n=0.
    \end{cases}
    \label{lorenzmodes}
\end{equation}
Fixing Dirichlet (Neumann) boundary conditions on $A_\mu$ corresponds to fixing the polarization vectors $\varepsilon_\mu^{(1)}$ ($\varepsilon_\mu^{(2)}$) up to a residual gauge transformation $\varepsilon_\mu^{(1)}\to \varepsilon_\mu^{(1)}+C_1k_\mu$ ($\varepsilon_\mu^{(2)}\to \varepsilon_\mu^{(2)}+C_2k_\mu$). The condition (\ref{lorenzmodes}) then implies that, for Dirichlet boundary conditions on $A_\mu$, all the $\varepsilon^{(2)}_z$ are also fixed up to a residual gauge transformation $\varepsilon_z^{(2)}\to\varepsilon_z^{(2)}+i\rho_n$. In other words, $\beta_z(x)$ (the subleading part of $A_z$ at the boundary) is fixed: Neumann boundary conditions are imposed on $A_z$.

On the other hand, for Neumann boundary conditions on $A_\mu$, all the $\varepsilon_z^{(1)}$ with $n\geq 1$ are fixed up to a gauge transformation by the condition (\ref{lorenzmodes}), while $\varepsilon_z^{(1)}(k_x,k_y,\rho_0)$ is unconstrained. However, it is easy to show that, choosing the correct gauge-invariant boundary condition (i.e. fixing $F_{\mu z}|_{z=0,\ell}$ instead of $\beta_\mu(x)$ only), $\varepsilon_z^{(1)}(k_x,k_y,\rho_0)$ is also constrained. Therefore, the second choice of boundary condition is given by Neumann on $A_\mu$ and Dirichlet on $A_z$. 

Now that we have specified an appropriate set of boundary conditions at the two boundaries, we can proceed to the quantization of the bulk gauge field. In the following, we will keep referring to Dirichlet and Neumann boundary conditions as determined by the boundary condition imposed on $A_\mu$. 

\subsection{Dirichlet boundary conditions: standard quantization}

For Dirichlet boundary conditions on $A_\mu$ and Neumann on $A_z$, the field we must quantize is given by
\begin{equation}
    A_\mu^D(x,z)=A_\mu^{(2)}(x,z), \hspace{2cm} A_z^D(x,z)=A_z^{(1)}(x,z)
\end{equation}
where $A_I^{(1,2)}$ are defined in equation (\ref{expansions}). The Lorenz condition does not fix the gauge completely. In particular, we have a residual gauge freedom
\begin{equation}
    A_I^D\to A_I^D+\partial_I \lambda_D
\end{equation}
where we can write
\begin{equation}
    \lambda_D=\sum_{n=1}^\infty \int \frac{dk_xdk_y}{2\pi\sqrt{\ell}}C_D(k_x,k_y,\rho_n)\sin(\rho_n z)\textrm{e}^{-i(\omega_n t-k_xx-k_yy)}
\end{equation}
where $C_D(k_x,k_y,\rho_n)$ can be chosen arbitrarily. We can use such residual gauge freedom to set $\varepsilon_z^{(1)}(k_x,k_y,\rho_n)=0$ for all $n\geq 1$. Note that $\varepsilon_z^{(1)}(k_x,k_y,\rho_0)$ cannot be set to zero by such a residual gauge transformation. Normalization for the resulting 3D massless scalar field imposes $\varepsilon_z^{(1)}(k_x,k_y,\rho_0)=1$. 
The Lorenz gauge condition then reduces to $\varepsilon_\mu^{(2)}k^\mu=0$, which removes one degree of freedom in the $\mu$ components of the gauge field. We can define an orthonormal basis of polarization 3-vectors $\{\epsilon_\mu^{(j)}\}|_{j=a,b}$ satisfying such condition:
\begin{equation}
    \begin{aligned}
        &\tilde{\epsilon}_\mu^{(a)}(k,\rho_n)=\frac{1}{\sqrt{k_x^2+k_y^2}}(0,-k_y,k_x)\\
        &\tilde{\epsilon}_\mu^{(b)}(k,\rho_n)=\frac{1}{\rho_n\sqrt{k_x^2+k_y^2}}(-(k_x^2+k_y^2),\omega_n k_x,\omega_n k_y)
    \end{aligned}
    \label{polbasis}
\end{equation}
with $n\geq 1$. We are now left with only physical degrees of freedom and can finally quantize the field:
\begin{equation}
    \begin{aligned}
        &\hat{A}_\mu^D=\sum_{n=1}^\infty \int \frac{dk_xdk_y}{2\pi\sqrt{\omega_n}\sqrt{\ell}}\sum_{j=a,b}\left[\tilde{\varepsilon}_\mu^{(j)}(k,\rho_n)\sin(\rho_n z)\textrm{e}^{-(i\omega_n t-k_xx-k_yy)}\hat{a}^{(j)}_{k,\rho_n}+h.c.\right]\\
        &\hat{A}_z^D=\int\frac{dk_xdk_y}{2\pi\sqrt{2\omega_0}\sqrt{\ell}}\left[\textrm{e}^{-(i\omega_0 t-k_xx-k_yy)}\hat{a}_{k}+h.c.\right]
    \end{aligned}
\end{equation}
where the creation and annihilation operators $\hat{a}$, $\hat{a}^\dagger$ satisfy the canonical commutation relations. For the physical interpretation of this result, see Section \ref{sec:gauge_field}.

\subsection{Neumann/electric boundary conditions: alternative quantization}

For Neumann boundary conditions, where we fix $F_{\mu z}|_{z=0,\ell}$, the field we must quantize is
\begin{equation}
    A_\mu^N(x,z)=A_\mu^{(1)}(x,z), \hspace{2cm} A_z^N(x,z)=A_z^{(2)}(x,z).
\end{equation}
The residual gauge freedom is now given by
\begin{equation}
    A_I^N\to A_I^N +\partial_I \lambda_N
\end{equation}
with 
\begin{equation}
    \lambda_N=\sum_{n=0}^\infty \int \frac{dk_xdk_y}{2\pi\sqrt{\ell}}C_N(k_x,k_y,\rho_n)\cos(\rho_n z)\textrm{e}^{-i(\omega_n t-k_xx-k_yy)}.
\end{equation}
Using this residual gauge freedom, we can set $\varepsilon^{(2)}_z(k,\rho_n)=0$ for all $\rho_n\geq 1$, and $\varepsilon^{(1)}_0(k,\rho_0)=0$. After fixing the residual gauge, the Lorenz condition reduces to $\varepsilon^{(1)}_\mu k^\mu=0$. For all the modes with $n\geq 1$, we can use the basis of polarization vectors (\ref{polbasis}). Now we also have to define an additional polarization vector for the $n=0$ mode of $A_\mu^N$. Having set $\varepsilon^{(1)}_0(k,\rho_0)=0$, the Lorenz gauge condition for this mode reads $k^i\varepsilon^{(1)}_i(k,\rho_0)=0$ with $i=1,2$. The unique unit vector satisfying such condition is
\begin{equation}
    \tilde{\varepsilon}^{(0)}_\mu(k)=\frac{1}{\sqrt{k_x^2+k_y^2}}(0,-k_y,k_x).
\end{equation}
Now that all the non-physical degrees of freedom have been eliminated, we can quantize the field:
\begin{equation}
\begin{aligned}
        &\hat{A}_\mu^N=\int \frac{dk_xdk_y}{2\pi\sqrt{2\omega_0}\sqrt{\ell}}\left[\tilde{\varepsilon}_\mu^{(0)}(k)\textrm{e}^{-(i\omega_0 t-k_xx-k_yy)}\hat{a}^{(j)}_{k}+h.c.\right]\\
        &+\sum_{n=1}^\infty \int \frac{dk_xdk_y}{2\pi\sqrt{\omega_n}\sqrt{\ell}}\sum_{j=a,b}\left[\tilde{\varepsilon}_\mu^{(j)}(k,\rho_n)\cos(\rho_n z)\textrm{e}^{-(i\omega_n t-k_xx-k_yy)}\hat{a}^{(j)}_{k,\rho_n}+h.c.\right].
\end{aligned}
\end{equation}
For the physical interpretation of this result, see Section \ref{sec:gauge_field}.

\section{Derivation of the asymptotic spectrum}\label{app:asymp_spectrum}

Consider the Schr\"odinger equation with the ``test'' potential
\begin{equation}
    \Vt(z) \equiv
    \frac{\alpha}{(z_0-|z|)^2}+\frac{\beta}{z_0-|z|}+c
    \label{test}
\end{equation}
The potential (\ref{test}) is simple enough that we can explicitly compute the asymptotic form of its normalizable spectrum, $\lamt_j$. Indeed, a simple rescaling of $z$ reduces the corresponding Schr\"odinger equation to a Whittaker's differential equation, whose solution is well known \cite{colton_1970}. First, we note the basic result from Sturm-Liouville theory that since the potential $\Vt(z)$ is even, the normalizable eigenfunctions alternate between even functions ($u_0,u_2,\dots$) and odd functions ($u_1,u_3,\dots$). Let us consider the odd eigenfunctions, which are defined, up to normalization, by the condition $u(0)=0$. The solution to the Schr\"odinger equation with this ``initial" condition at $z=0$ can be obtained from the solution of the Whittaker's differential equation by an appropriate rescaling of $z$ \cite{colton_1970}. The general solution is given by the linear combination
\begin{equation}
    u_\text{odd}(z,\lamt)= C_1 M_{x,\frac{\mu}{2}}\left(i2(z_0-|z|)\sqrt{\tilde{\lambda}-c}\right) + C_2 M_{x,-\frac{\mu}{2}}\left(i2(z_0-|z|)\sqrt{\tilde{\lambda}-c}\right)
    \label{whittakersol}
\end{equation}
where $M_{x,\frac{\mu}{2}}(y)$ is related to the confluent hypergeometric function of the first kind ${}_1F_1(a;b;y)$ by
\begin{equation}
    M_{x,\frac{\mu}{2}}(y)=y^{\frac{1+\mu}{2}}\textrm{e}^{-\frac{y}{2}}{}_1F_1\left(\frac{1+\mu}{2}-x;1+\mu;y\right)
\end{equation}
and we identified
\begin{equation}
    x=\frac{i\beta}{2\sqrt{\tilde{\lambda}-c}};\hspace{1.5cm} \mu=2\bar{\Delta};\hspace{1.5cm} y=i2(z_0-|z|)\sqrt{\tilde{\lambda}-c}
    \label{hyperparameters}
\end{equation}
where $\bar{\Delta}=\sqrt{1+4\alpha}/2$. The first term in the linear combination (\ref{whittakersol}) corresponds to the normalizable component of the eigenfunctions, while the second term is the non-normalizable component. Since we are interested in the normalizable spectrum of the potential (\ref{test}), we set $C_2=0$.

Then the odd normalizable eigenvalues $\lamt_j$, $j=1,3,5...$, are given by values of $\lamt$ that solve the odd condition $u_{odd}(0,\tilde{\lambda}_j)=0$.\footnote{For $-1/4<a<0$ the non-normalizable modes also vanish at $z=z_0$. To avoid this difficulty, here we are assuming that $a>0$. The results for the case $a<0$ follow by analytic continuation.} Since we are interested in the asymptotic spectrum, we want to look for solutions $\{\tilde{\lambda}_j\}$ of the odd condition with $\tilde{\lambda}_j\gg 1$.  Therefore, we can set $z=0$ and find the form of the normalizable eigenfunction (first term in equation (\ref{whittakersol})) in the limit of large $\tilde{\lambda}$. Defining $x=i\tau$, $y=i\zeta$ and assuming $\tau\in \mathbb{R}^+$, $\zeta\in\mathbb{R}\backslash \{0\}$, $\mu\in\mathbb{R}$, the function $M_{i\tau,\frac{\mu}{2}}(i\zeta)$ has the following asymptotic expansion for large $\zeta$ \cite{colton_1970}:
\begin{equation}
\begin{aligned}
    M_{i\tau,\frac{\mu}{2}}&(i\zeta)\sim \frac{2\Gamma(1+\mu)\exp\left[\frac{\pi}{2}\left(\tau+i\frac{1+\mu}{2}\right)\right]\textrm{sgn}(\zeta)}{\left|\Gamma\left(\frac{1+\mu}{2}\pm i\tau\right)\right|}\cdot\\[15pt]
    &\left\{\cos\left[-\tau\log|\zeta|+\frac{\zeta}{2}+\delta-\frac{\pi}{4}(1+\mu)\textrm{sgn}(\zeta) \right]\sum_{k=0}^N \frac{\left(\frac{\tau^2}{\zeta}\right)^k}{k!}\frac{\cos\left[\sum_{r=1}^k(\phi_r^{(+)}+\phi_r^{(-)})-\frac{\pi}{2} k\right]}{\prod_{r=1}^k\left(\sin\phi_r^{(+)}\sin\phi_r^{(-)}\right)}\right.\\[15pt]
    &\left.-\sin\left[-\tau\log|\zeta|+\frac{\zeta}{2}+\delta-\frac{\pi}{4}(1+\mu)\textrm{sgn}(\zeta) \right]\sum_{k=1}^N \frac{\left(\frac{\tau^2}{\zeta}\right)^k}{k!}\frac{\sin\left[\sum_{r=1}^k(\phi_r^{(+)}+\phi_r^{(-)})-\frac{\pi}{2} k\right]}{\prod_{r=1}^k\left(\sin\phi_r^{(+)}\sin\phi_r^{(-)}\right)}\right\}
    \label{asymptexp}
\end{aligned}
\end{equation}
where $\Gamma(x)$ is the Euler gamma function, and $\delta$, $\phi_r^{(\pm)}$ are given by
\begin{equation}
    \delta=-\frac{i}{2}\log\left[\frac{\Gamma\left(\frac{1+\mu}{2}+ i\tau\right)}{\Gamma\left(\frac{1+\mu}{2}- i\tau\right)}\right]; \hspace{1.5cm} \tan\phi^{(\pm)}_r=\frac{\tau}{r+\frac{-1\pm\mu}{2}}.
\end{equation}
Keeping only terms up to order $1/\sqrt{\tilde{\lambda}-c}$ in equation (\ref{asymptexp}), we get (up to an irrelevant multiplicative factor)
\begin{equation}
\begin{aligned}
    M_{i\tau,\frac{\mu}{2}}(i\zeta)\sim & \cos\left[-\tau\log|\zeta|+\frac{\zeta}{2}+\delta-\frac{\pi}{4}(1+\mu)\textrm{sgn}(\zeta) \right]\\[15pt]
    &+\frac{1-\mu^2}{4\zeta}\sin\left[-\tau\log|\zeta|+\frac{\zeta}{2}+\delta-\frac{\pi}{4}(1+\mu)\textrm{sgn}(\zeta) \right]+\mathcal{O}\left(\frac{1}{\tilde{\lambda}-c}\right).
\end{aligned}
\label{asymptexp2}
\end{equation}
The discrete set of asymptotic odd eigenvalues $\{\tilde{\lambda}_j\}$ (for odd values of $j\gg 1$) we are looking for is determined by the zeros of equation (\ref{asymptexp2}). Using equation (\ref{hyperparameters}), such zeros are given by the $\tilde{\lambda}$'s satisfying
\begin{equation}
\begin{aligned}
    \sqrt{\tilde{\lambda}-c}=&-\frac{\alpha}{2z_0^2\sqrt{\tilde{\lambda}-c}}+\frac{\beta}{2z_0\sqrt{\tilde{\lambda}-c}}\log\left(2z_0\sqrt{\tilde{\lambda}-c}\right)+\frac{i}{2z_0}\log\left[\frac{\Gamma\left(\Delta_++\frac{i\beta}{2\sqrt{\tilde{\lambda}-c}}\right)}{\Gamma\left(\Delta_+-\frac{i\beta}{2\sqrt{\tilde{\lambda}-c}}\right)}\right]\\
    &+\frac{\pi}{2z_0}(\Delta_++2n-1)
    \end{aligned}
    \label{zeros}
\end{equation}
with $\Delta_+=(1+2\bar{\Delta})/2$ and $n=1,2,3...$. Substituting the expansion ansatz
\begin{equation}
    \sqrt{\tilde{\lambda}-c}=a_1n+a_2\log n+a_3+a_4\frac{\log n}{n}+\frac{a_5}{n}+\mathcal{O}\left(\frac{1}{n^2}\right)
    \label{ansatzlambda}
\end{equation}
into equation (\ref{zeros}) and expanding the right hand side up to order $1/n$ yields the values of the $\{a_i\}_{i=1,...,5}$. Solving equation (\ref{ansatzlambda}) for $\tilde{\lambda}$ and identifying $n=(j+1)/2$ with $j$ odd, we finally obtain the expression for the asymptotic spectrum:
\begin{align}
    \lamt_j\sim &
    \left(\frac{\pi}{2z_0}\right)^2 j^2+
    \left(\frac{\pi}{2z_0}\right)^2 2\Delta_+ j+
    \frac{\beta}{z_0}\log(j)\notag\\
    &
    +\left(\frac{\pi}{2z_0}\right)^2\left(\Delta_+^2-\frac{4\alpha}{\pi^2}\right)
    +\frac{\beta}{z_0}\left[\log\left(\pi\right)-\psi(\Delta_+)\right]+c,
\end{align}
where $\psi(x)$ is the digamma function. Although we have derived this expansion for $j$ odd, it can immediately be extended also to even values of $j$, yielding the full asymptotic spectrum (\ref{eq:fitformula}).

\section{Proof of lemma}\label{app:lemma}

Here we prove the lemma used in the reconstruction algorithm in Section \ref{sec:inverse_SL_problem}.

\begin{lemma} 
For any $L^2$ integrable function $f(z)$ on $z\in(-z_0,z_0)$ the following identity holds:
\begin{align}\label{eq:modified_SL_expansionapp}
    f(z)
    =
    \sum_{j=0}^\infty
    \frac{\vt_j(z)\int_{-z_0}^z dy\,u_j(y) f(y)+\ut_j(z)\int_{z}^{z_0} dy\,v_j(y) f(y)}{\omega'(\lam_j)}.
\end{align}
\end{lemma}

\begin{proof}
Let $f(z)$ be an $L^2$ integrable function. We define
\begin{align}
    \Phi(z,\lam)
    \equiv
    \frac{\vt(z,\lam)\int_{-z_0}^z dy\,u(y,\lam) f(y)+\ut(z,\lam)\int_{z}^{z_0} dy\,v(y,\lam) f(y)}{\omega(\lam)}.
\end{align}
Notice that this is well-defined because $u(y,\lam)$ is integrable at the left boundary, while $v(y,\lam)$ is integrable at the right boundary. Now consider the contour integral
\begin{align}
    \frac{1}{2\pi i}\int_\Gamma d\lam\, \Phi(z,\lam),
\end{align}
where $\Gamma$ is a large circle in the complex $\lam$ plane, whose radius we take to infinity. In this limit, we can evaluate the contour integral in two different ways. First, directly by the residue theorem we find
\begin{align}\label{eq:Phi_integral_1}
    \frac{1}{2\pi i}\int_\Gamma d\lam\, \Phi(z,\lam)
    =
    \sum_{j=0}^\infty
    \frac{\vt_j(z)\int_{-z_0}^z dy\,u_j(y) f(y)+\ut_j(z)\int_{z}^{z_0} dy\,v_j(y) f(y)}{\omega'(\lam_j)},
\end{align}
where $\lam_j$ are eigenvalues of \eqref{eq:SE2} with normalizable boundary conditions, and where we define $u_j(z)\equiv u(z,\lam_j)$, $\ut_j(z)\equiv \ut(z,\lam_j)$, and similarly for the $v$s.

Alternatively, we can evaluate the above contour integral by first going to the limit of large $|\lam|$ (since we are taking the radius of the contour $\Gamma$ to infinity) before using the residue theorem. As we discussed, in the limit $|\lam|\rightarrow\infty$ we can effectively replace $V(z)$ with $\Vt(z)$, and hence $u\rightarrow \ut$,  $v\rightarrow\vt$ and $\omega(\lam)\rightarrow\tilde\omega(\lam)\equiv-\utm(z_0,\lam)$. Using the residue theorem we thus obtain
\begin{align}\label{eq:contour}
    \frac{1}{2\pi i}\int_\Gamma d\lam\, \Phi(z,\lam)
    =
    \sum_{j=0}^\infty
    \frac{\Vt_j(z)\int_{-z_0}^z dy\,\Ut_j(y) f(y)+\Ut_j(z)\int_{z}^{z_0} dy\,\Vt_j(y) f(y)}{\tilde\omega'(\lamt_j)},
\end{align}
where $\Ut_j$ and $\Vt_j$ are solutions of the following initial value problems on $z\in(-z_0,z_0)$ 
\begin{align} 
\Ut_j:
    \begin{cases}
    \Ut_j''(z)=\big[\Vt(z)-\lamt_j\big]\Ut_j(z)\\
    \Ut_j^{(+)}(-z_0)=\frac{1}{2\Dbar}\\
    \Ut_j^{(-)}(-z_0)=0
    \end{cases}
\Vt_j:
    \begin{cases}
    \Vt_j''(z)=\big[\Vt(z)-\lamt_j\big]v_j(z)\\
    \Vt_j^{(+)}(+z_0)=\frac{1}{2\Dbar}\\
    \Vt_j^{(-)}(+z_0)=0
    \end{cases}.
\end{align}
Hence $\Ut_j$ and $\Vt_j$ must both be proportional to the $j$-th normalizable eigenfunction of $u''(z)=\big[\Vt(z)-\lam\big]u(z)$. Since $\Vt(z)$ is $\mathbb Z_2$ symmetric, these eigenfunctions have a parity $(-1)^j$, and thus $\Ut_j(z)=(-1)^j\Vt_j(z)$. Additionally, by using the explicit form \ref{whittakersol} for the normalizable eigenfunctions, it can be verified by direct calculation that $\tilde\omega'(\lamt_j)=(-1)^j||U_j||^2$, where $||\cdot||^2$ denotes the $L^2$-norm, i.e. the norm associated with the Sturm-Liouville inner product. It was in order to obtain this identity that we chose the normalization factor $1/2\bar\Delta$ for the normalizable modes in Eqs. \eqref{eq:u_and_v} and \eqref{eq:ut_and_vt}. Therefore \eqref{eq:contour} gives
\begin{align}
    \frac{1}{2\pi i}\int_\Gamma d\lam\, \Phi(z,\lam)
    =
    \sum_{j=0}^\infty
    \frac{\Ut_j(z)\int_{-z_0}^{z_0} dy\,\Ut_j(y) f(y)}{||U_j||^2},
\end{align}
which is precisely the Sturm-Liouville expansion of $f(z)$. Comparing with \eqref{eq:Phi_integral_1} we obtain the identity
\begin{align}
    f(z)
    =
    \sum_{j=0}^\infty
    \frac{\vt_j(z)\int_{-z_0}^z dy\,u_j(y) f(y)+\ut_j(z)\int_{z}^{z_0} dy\,v_j(y) f(y)}{\omega'(\lam_j)},
\end{align}
which completes the proof.
\end{proof}

\newpage

\bibliographystyle{jhep}
\bibliography{references}

\providecommand{\href}[2]{#2}\begingroup\raggedright\begin{thebibliography}{10}

\bibitem{VanRaamsdonk:2021qgv}
M.~Van~Raamsdonk, \emph{Cosmology from confinement?}, {\emph{Journal of High
  Energy Physics} {\bf 2022} (2022) 1--39}.

\bibitem{Antonini2022}
S.~Antonini, P.~Simidzija, B.~Swingle and M.~Van~Raamsdonk, \emph{{Cosmology
  from the vacuum}},  \href{https://arxiv.org/abs/2203.11220}{{\tt
  2203.11220}}.

\bibitem{Antonini2022short}
S.~Antonini, P.~Simidzija, B.~Swingle and M.~Van~Raamsdonk, \emph{{Cosmology as
  a holographic wormhole}},  \href{https://arxiv.org/abs/2206.14821}{{\tt
  2206.14821}}.

\bibitem{Maldacena:2004rf}
J.~M. Maldacena and L.~Maoz, \emph{{Wormholes in AdS}},
  \href{http://dx.doi.org/10.1088/1126-6708/2004/02/053}{\emph{JHEP} {\bf 02}
  (2004) 053}, [\href{https://arxiv.org/abs/hep-th/0401024}{{\tt
  hep-th/0401024}}].

\bibitem{Freivogel:2019lej}
B.~Freivogel, V.~Godet, E.~Morvan, J.~F. Pedraza and A.~Rotundo, \emph{{Lessons
  on Eternal Traversable Wormholes in AdS}},
  \href{http://dx.doi.org/10.1007/JHEP07(2019)122}{\emph{JHEP} {\bf 07} (2019)
  122}, [\href{https://arxiv.org/abs/1903.05732}{{\tt 1903.05732}}].

\bibitem{May:2021xhz}
A.~May, P.~Simidzija and M.~Van~Raamsdonk, \emph{Negative energy enhancement in
  layered holographic conformal field theories}, {\emph{Journal of High Energy
  Physics} {\bf 2021} (2021) 1--24}.

\bibitem{Gao2016}
P.~Gao, D.~L. Jafferis and A.~Wall, \emph{{Traversable Wormholes via a Double
  Trace Deformation}},
  \href{http://dx.doi.org/10.1007/JHEP12(2017)151}{\emph{JHEP} {\bf 12} (2017)
  151}, [\href{https://arxiv.org/abs/1608.05687}{{\tt 1608.05687}}].

\bibitem{Maldacena:2017axo}
J.~Maldacena, D.~Stanford and Z.~Yang, \emph{{Diving into traversable
  wormholes}}, \href{http://dx.doi.org/10.1002/prop.201700034}{\emph{Fortsch.
  Phys.} {\bf 65} (2017) 1700034},
  [\href{https://arxiv.org/abs/1704.05333}{{\tt 1704.05333}}].

\bibitem{Maldacena:2018lmt}
J.~Maldacena and X.-L. Qi, \emph{{Eternal traversable wormhole}},
  \href{https://arxiv.org/abs/1804.00491}{{\tt 1804.00491}}.

\bibitem{Maldacena:2018gjk}
J.~Maldacena, A.~Milekhin and F.~Popov, \emph{{Traversable wormholes in four
  dimensions}},  \href{https://arxiv.org/abs/1807.04726}{{\tt 1807.04726}}.

\bibitem{Witten1998a}
E.~Witten, \emph{{Anti-de Sitter space, thermal phase transition, and
  confinement in gauge theories}},
  \href{http://dx.doi.org/10.4310/ATMP.1998.v2.n3.a3}{\emph{Adv. Theor. Math.
  Phys.} {\bf 2} (1998) 505--532},
  [\href{https://arxiv.org/abs/hep-th/9803131}{{\tt hep-th/9803131}}].

\bibitem{Sakai:2004cn}
T.~Sakai and S.~Sugimoto, \emph{{Low energy hadron physics in holographic
  QCD}}, \href{http://dx.doi.org/10.1143/PTP.113.843}{\emph{Prog. Theor. Phys.}
  {\bf 113} (2005) 843--882}, [\href{https://arxiv.org/abs/hep-th/0412141}{{\tt
  hep-th/0412141}}].

\bibitem{Levitan-Gasymov}
B.~M. {Levitan} and M.~G. {Gasymov}, \emph{{Determination of a Differential
  Equation by Two of its Spectra}},
  \href{http://dx.doi.org/10.1070/RM1964v019n02ABEH001145}{\emph{Russian
  Mathematical Surveys} {\bf 19} (Apr., 1964) R01}.

\bibitem{Hald}
O.~H. Hald, \emph{{The inverse Sturm-Liouville problem with symmetric
  potentials}}, \href{http://dx.doi.org/10.1007/BF02545749}{\emph{Acta
  Mathematica} {\bf 141} (1978) 263 -- 291}.

\bibitem{kac1966}
M.~Kac, \emph{Can one hear the shape of a drum?}, {\emph{The american
  mathematical monthly} {\bf 73} (1966) 1--23}.

\bibitem{Faulkner:2018faa}
T.~Faulkner, M.~Li and H.~Wang, \emph{{A modular toolkit for bulk
  reconstruction}},
  \href{http://dx.doi.org/10.1007/JHEP04(2019)119}{\emph{JHEP} {\bf 04} (2019)
  119}, [\href{https://arxiv.org/abs/1806.10560}{{\tt 1806.10560}}].

\bibitem{Ryu2006}
S.~Ryu and T.~Takayanagi, \emph{Holographic derivation of entanglement entropy
  from the anti--de sitter space/conformal field theory correspondence},
  {\emph{Physical review letters} {\bf 96} (2006) 181602}.

\bibitem{Cooper2018}
S.~Cooper, M.~Rozali, B.~Swingle, M.~Van~Raamsdonk, C.~Waddell and D.~Wakeham,
  \emph{Black hole microstate cosmology}, {\emph{Journal of High Energy
  Physics} {\bf 2019} (2019) 1--70}.

\bibitem{Antonini2019}
S.~Antonini and B.~Swingle, \emph{{Cosmology at the end of the world}},
  \href{http://dx.doi.org/10.1038/s41567-020-0909-6}{\emph{Nature Phys.} {\bf
  16} (2020) 881--886}, [\href{https://arxiv.org/abs/1907.06667}{{\tt
  1907.06667}}].

\bibitem{Antonini:2021xar}
S.~Antonini and B.~Swingle, \emph{{Holographic boundary states and
  dimensionally reduced braneworld spacetimes}},
  \href{http://dx.doi.org/10.1103/PhysRevD.104.046023}{\emph{Phys. Rev. D} {\bf
  104} (2021) 046023}, [\href{https://arxiv.org/abs/2105.02912}{{\tt
  2105.02912}}].

\bibitem{Almheiri:2018ijj}
A.~Almheiri, A.~Mousatov and M.~Shyani, \emph{{Escaping the Interiors of Pure
  Boundary-State Black Holes}},  \href{https://arxiv.org/abs/1803.04434}{{\tt
  1803.04434}}.

\bibitem{Fallows:2022ioc}
S.~Fallows and S.~F. Ross, \emph{{Constraints on cosmologies inside black
  holes}},  \href{https://arxiv.org/abs/2203.02523}{{\tt 2203.02523}}.

\bibitem{Waddell:2022fbn}
C.~Waddell, \emph{{Bottom-up holographic models for cosmology}},
  \href{https://arxiv.org/abs/2203.03096}{{\tt 2203.03096}}.

\bibitem{gradshteyn2014table}
I.~S. Gradshteyn and I.~M. Ryzhik, \emph{Table of integrals, series, and
  products}.
\newblock Academic press, 2014.

\bibitem{Aharony:2010ay}
O.~Aharony, D.~Marolf and M.~Rangamani, \emph{{Conformal field theories in
  anti-de Sitter space}},
  \href{http://dx.doi.org/10.1007/JHEP02(2011)041}{\emph{JHEP} {\bf 02} (2011)
  041}, [\href{https://arxiv.org/abs/1011.6144}{{\tt 1011.6144}}].

\bibitem{Hijano:2020szl}
E.~Hijano and D.~Neuenfeld, \emph{{Soft photon theorems from CFT Ward identites
  in the flat limit of AdS/CFT}},
  \href{http://dx.doi.org/10.1007/JHEP11(2020)009}{\emph{JHEP} {\bf 11} (2020)
  009}, [\href{https://arxiv.org/abs/2005.03667}{{\tt 2005.03667}}].

\bibitem{Witten:2003ya}
E.~Witten, \emph{{SL(2,Z) action on three-dimensional conformal field theories
  with Abelian symmetry}},  in \emph{{From Fields to Strings: Circumnavigating
  Theoretical Physics: A Conference in Tribute to Ian Kogan}}, pp.~1173--1200,
  7, 2003.
\newblock \href{https://arxiv.org/abs/hep-th/0307041}{{\tt hep-th/0307041}}.

\bibitem{Yee:2004ju}
H.-U. Yee, \emph{{A Note on AdS / CFT dual of SL(2,Z) action on 3-D conformal
  field theories with U(1) symmetry}},
  \href{http://dx.doi.org/10.1016/j.physletb.2004.05.082}{\emph{Phys. Lett. B}
  {\bf 598} (2004) 139--148}, [\href{https://arxiv.org/abs/hep-th/0402115}{{\tt
  hep-th/0402115}}].

\bibitem{Marolf:2006nd}
D.~Marolf and S.~F. Ross, \emph{{Boundary Conditions and New Dualities: Vector
  Fields in AdS/CFT}},
  \href{http://dx.doi.org/10.1088/1126-6708/2006/11/085}{\emph{JHEP} {\bf 11}
  (2006) 085}, [\href{https://arxiv.org/abs/hep-th/0606113}{{\tt
  hep-th/0606113}}].

\bibitem{breitenlohner1982stability}
P.~Breitenlohner and D.~Z. Freedman, \emph{Stability in gauged extended
  supergravity}, {\emph{Annals of physics} {\bf 144} (1982) 249--281}.

\bibitem{breitenlohner1982positive}
P.~Breitenlohner and D.~Z. Freedman, \emph{Positive energy in anti-de sitter
  backgrounds and gauged extended supergravity}, {\emph{Physics letters B} {\bf
  115} (1982) 197--201}.

\bibitem{Bilson:2010ff}
S.~Bilson, \emph{{Extracting Spacetimes using the AdS/CFT Conjecture: Part
  II}}, \href{http://dx.doi.org/10.1007/JHEP02(2011)050}{\emph{JHEP} {\bf 02}
  (2011) 050}, [\href{https://arxiv.org/abs/1012.1812}{{\tt 1012.1812}}].

\bibitem{Hartman2013a}
T.~Hartman and J.~Maldacena, \emph{{Time evolution of entanglement entropy from
  black hole interiors}},
  \href{http://dx.doi.org/10.1007/JHEP05(2013)014}{\emph{Journal of High Energy
  Physics} {\bf 2013} (2013) }, [\href{https://arxiv.org/abs/1303.1080}{{\tt
  1303.1080}}].

\bibitem{colton_1970}
D.~Colton, \emph{The confluent hypergeometric function. by herbert buchholz.
  springer-verlag, new york (1969). 238 pp.},
  \href{http://dx.doi.org/10.1017/S0008439500031064}{\emph{Canadian
  Mathematical Bulletin} {\bf 13} (1970) 164–164}.

\end{thebibliography}\endgroup

\end{document}